\documentclass[sigconf]{acmart}

\setcopyright{acmlicensed}
\acmYear{2027}
\copyrightyear{2027}
\acmConference[SIGMOD '27]{SIGMOD Conference on Management of Data}
  {June 13--19, 2027}{Huntington Beach, CA, USA}
\acmDOI{XXXXXXX.XXXXXXX}
\acmISBN{978-x-xxxx-xxxx-x/YYYY/MM}

\usepackage{amsmath}
\usepackage{bm}
\usepackage{graphicx}
\usepackage{etoolbox}
\usepackage{enumitem}
\usepackage{algorithmic}
\usepackage{multirow}
\usepackage{amsthm}
\usepackage{booktabs}
\usepackage{xspace}
\usepackage{xcolor}

\usepackage{bbold}
\usepackage{multirow}
\usepackage{diagbox}
\usepackage{balance}
\usepackage{float}
\usepackage[noend,lined,boxed,vlined,ruled,linesnumbered]
{algorithm2e}
\usepackage{makecell}

\usepackage{array}
\usepackage{readme-method}
\usepackage{readme-results}
\usepackage{pgfplots}
\usepgfplotslibrary{groupplots}
\usetikzlibrary{patterns,arrows.meta}
\pgfplotsset{compat=1.18}

\newcolumntype{C}[1]{>{\centering\arraybackslash}p{#1}}

\definecolor{myyellow}{HTML}{fcefd2}
\definecolor{mydarkyellow}{HTML}{FAD02E}

\definecolor{mygreen}{HTML}{d2fce6}
\definecolor{mypurple}{HTML}{f2d2fc}
\definecolor{peach}{HTML}{FFDAB9}
\definecolor{fpred}{RGB}{255,99,99}

\usepackage[labelformat=simple]{subcaption}

\SetKwInOut{Input}{input}
\SetKwInOut{Output}{output}

\DeclareMathOperator*{\argmax}{arg\,max}

\newcommand{\code}[1]{\texttt{\small #1}}

\usepackage{booktabs}
\usepackage{multirow}

\newcommand{\ignore}[1]{}

\newtheorem{df}{Definition}
\newtheorem{ex}{Example}
\newtheorem{pr}{Proposition}
\newtheorem{lemm}{Lemma}
\newtheorem{lem}{Theorem}

\newenvironment{example}{\begin{ex} \nopagebreak
\begin{rm}}{{\hfill$\Box$}\end{rm}\end{ex}} 

\newenvironment{definition}{\begin{df} \nopagebreak
\begin{rm}}{{}\end{rm}\end{df}}

\newcommand{\stitle}[1]{\vspace{1ex}\noindent{\bf #1}}

\usepackage{etoolbox}

\newtoggle{full}
\toggletrue{full}

\begin{document}


\title{\method: \\ Documenting Table Relationships for Diverse Data Tasks}

\author{Jiaming Liang}
\authornote{Part of work done while at Microsoft.}
\affiliation{
  \institution{University of Pennsylvania}
  \city{Philadelphia}
  \state{Pennsylvania}
  \country{USA}
}

\author{Yeye He}
\authornote{Correspondence: liangjm@seas.upenn.edu, yeyehe@microsoft.com}
\affiliation{
  \institution{Microsoft Research}
  \city{Redmond}
  \state{Washington}
  \country{USA}
}

\author{Haoyu Dong}
\affiliation{
  \institution{Microsoft Research}
  \city{Redmond}
  \state{Washington}
  \country{USA}
}

\author{Zachary Ives}
\affiliation{
  \institution{University of Pennsylvania}
  \city{Philadelphia}
  \state{Pennsylvania}
  \country{USA}
}

\author{Surajit Chaudhuri}
\affiliation{
  \institution{Microsoft Research}
  \city{Redmond}
  \state{Washington}
  \country{USA}
}



\begin{abstract}
Large language models (LLMs) can perform diverse data tasks when given a small set of relevant tables. However, for complex tasks requiring multiple tables over large data lakes, identifying relevant tables for the task becomes a prerequisite, but remains  challenging. This is because there are often many tables in a data lake with similar table and schema names, and limited additional semantic context, thereby forcing agents to spend substantial time exploring tables at online query time, greatly affecting models' ability to perform tasks accurately, while increasing cost/latency substantially.

We present \method, a novel table documentation approach that uses lightweight offline processing to infer and index rich \emph{table-to-table relationships}, including confusable, joinable, and union-able tables. 
We further develop a suite of agentic tools for LLMs to work with the index, to help agents retrieve relevant tables and perform tasks much more efficiently and effectively. Finally, we study how to optimize our index construction with a budget constraint. On data lakes containing up to 12K tables, \method improves task accuracy by as much as 40 percentage points across a range of data discovery, data analysis, and data science tasks,  outperforming diverse baselines, including powerful coding agents with unlimited time and monetary budgets. We also show that equipping existing methods (such as coding agents) with \method can not only lead to consistent quality gains, but also significantly reduce their latency by up to 6x, demonstrating the potential of \method.
\end{abstract}

\maketitle
\fancyhead{}

\section{Introduction}
\label{sec:intro}

Large language models (LLMs) are increasingly capable of performing diverse tasks over tabular data, when given a small set of tables relevant to a task \cite{table-gpt, narayan2022can}. However, users
often need to perform data discovery, data analysis, and data science over a
large data lake containing thousands of tables, without knowing in advance
which tables are relevant \cite{halevy2016goods, fernandez2018aurum}. Such tasks remain challenging for today's models and agentic harnesses \cite{kramabench, wang2026lakeqa}.

Consider a concrete example data lake tested in our experiments, illustrated in
Figure~\ref{fig:system-overview}. It contains 17 tables named
``\code{Sales},'' 29 named ``\code{Customers},'' and 21 named
``\code{Country},'' many with similar schemas. To answer a question such as
``\code{global sales by customer country}'' in a particular region and time period, a human
or LLM agent must determine not only which \code{Sales} table is relevant, but
also which \code{Customers} and \code{Country} tables can be ``paired/joined'' with that \code{Sales} table to answer the question, which is unfortunately non-trivial in a large data lake, as table and schema names alone provide little clue for LLM to identify relevant tables. 

Therefore, today existing methods (including powerful data agents and coding agents~\cite{githubcopilotcli, codexcli, claudecode}) are often forced to open and explore a large number of tables that may be relevant  at online query time,  (e.g., all 17 \code{Sales} tables and all 29
\code{Customers} tables) before locating the ones relevant for a task.
This approach is clearly (1)~\underline{time-consuming}, because users must wait while
the agent explores the data; (2)~\underline{less accurate}, because the agent must often decide using partial, shallow evidence collected at runtime; and
(3)~\underline{wasteful}, as each new question triggers another expensive
round of exploration from scratch whose findings are seldom reused.

\stitle{README for codebases vs.\ README for data lakes.}
We draw inspiration from software engineering, where humans and LLMs face
a similar readability/discoverability problem in complex codebases, which is commonly addressed with a
\texttt{README}~\cite{README} and, increasingly, an \texttt{AGENTS.md} for LLM
agents~\cite{Agentsmd}. Placed at the root of a repository, these readme documents
explain what the code does and how its components relate to each other with navigatable links, allowing developers
and agents to reuse knowledge rather than rediscover it from scratch.

This analogy motivates our core research question: \emph{Can a
README-equivalent for tables benefit data agents operating over data lakes,
just as READMEs benefit coding agents operating over codebases?}

\stitle{\method for tables: relationships beyond simple summaries.}
A natural first step in documenting a dataset is to prompt an LLM to summarize
each table and column in natural language, as used in recent work such as
AutoDDG~\cite{autoddg}. 

While such approaches are clearly useful, we observe that free-form descriptions of \emph{individual tables} are insufficient
for many complex data tasks, which often depend on \emph{table-to-table relationships}, and would benefit from a more \textit{structured and holistic} view of
the lake. This resembles how a codebase README would not only describe files in
isolation, but also use pointers and hyperlinks to connect related components. We illustrate this with an example.

\begin{figure*}[t]
\vspace{-18mm}
  \centering
  \includegraphics[width=0.86\textwidth]{figures/system_overview.pdf}
  \vspace{-4mm}
  \caption{Overview of \method. An example \method in (b), illustrated with
    abridged table-level entries, is the output of offline indexing (a)
    and the input to online agentic tools (c). An agent at online query time can utilize the provided agentic tools to retrieves tables relevant to a given query, follows documented relationships, and inspects raw tables.}
  \label{fig:system-overview}
  \vspace{-4mm}
\end{figure*}

\begin{example}
\label{ex:intro-sales}
Consider the data lake we encountered in our experiments, whose 17 \code{Sales.csv} tables are labeled
\code{Sales\_1} through \code{Sales\_17} in Figure~\ref{fig:system-overview}(b) for ease of reference. Prior work such as AutoDDG would perform a per-table prompt asking LLM to describe each table using natural language, and populate such descriptions (semantic meaning, grains, statistical profiles of columns, etc.) in the leftmost column labeled as ``\code{Basic info (B)}'' in
Figure~\ref{fig:system-overview}(b). This is already helpful for retrieval tasks, especially for tables and columns with abbreviated or cryptic names~\cite{autoddg}.

However, given a question such as $Q =$ ``\code{global sales by customer country since 2024}'',
an agent still cannot easily determine from \emph{independent per-table descriptions}
which of the 17 \code{Sales}, 29 \code{Customers}, and 21 \code{Country}
tables should be used to answer the question.

In this work, we argue that agents need richer \emph{cross-table relationships}, including
(1)~confusable tables, (2)~joinable tables, and (3)~union-able tables, to answer
complex questions over large lakes. We extract such relationships
offline through a lightweight LLM-powered indexing step, as illustrated in
Figure~\ref{fig:system-overview}(a).

Specifically, the ``\code{Confusable tables (C)}'' column in
Figure~\ref{fig:system-overview}(b) explains how similarly-named and seemingly-identical tables actually differ in content. For example, \code{Sales\_1} and \code{Sales\_5}
cover different geographic regions, \code{Sales\_2} and \code{Sales\_8} have
different grains, and \code{Sales\_17} describes stock-sales rather than retail-sales. We extract these distinctions through an offline analysis that compares
similar tables and identifies distinguishing features, enabling efficient online consumption.

The next two columns in the figure similarly record how a table can join with or be
unioned with other tables in the lake, also extracted through a deep but efficient offline process.

At online query time, we develop agentic tools that let LLMs navigate these rich cross-table relationships and identify relevant
tables, as illustrated in
Figure~\ref{fig:system-overview}(c).  For our example query $Q$, by following
the pointers in \method, an agent can accurately and efficiently identify
\{\code{Sales\_1}, \code{Sales\_14}, \code{Country\_1},
\code{Customer\_9}\} as the tables needed to answer $Q$.
\end{example}

In this work we study the \method framework using three example types of table-relationships (extensible to additional types): (1)~\emph{confusable tables}, which look similar but differ in grain,
geography, time period, aggregation level, or other properties;
(2)~\emph{joinable tables}, which have a meaningful semantic join rather than
an accidental overlap in key ranges; and (3)~\emph{union-able tables}, which
share a compatible schema. 

We aim to show that, at the framework level, plugging in different types of relationships can bring substantial quality and latency benefits for diverse agent tasks operating on large data lakes.  While we present efficient techniques for extracting these relationships that can scale to large lakes, our goal is not to establish new state-of-the-art methods for each individual relationship type, which has a literature of its own. Instead, the framework is designed to accommodate alternative methods, allowing individual components to be readily replaced as better techniques emerge.




\stitle{Technical challenges.}
Producing this documentation for tens of thousands of tables raises
challenges that a per-table annotation method like AutoDDG~\cite{autoddg} does not face.

\underline{Efficient extraction of relationships}: The number of possible
relationships grows quadratically with the number of tables, where processing using LLM calls can become expensive very quickly.  We develop extraction algorithms
that scale linearly with the number of tables, or even sub-linearly if there is a fixed cost budget.

\underline{Resolving semantic ambiguity}: At scale, 
determining whether two tables are meaningfully related (e.g., confusable or joinable) becomes increasingly challenging and ambiguous. We develop data-driven techniques to resolve such ambiguities while adopting a late-binding strategy that retains plausible relationships until query time, when the user's question can provide additional evidence for disambiguation.

\underline{Agentic tools for \method}: In a large lake, even concise
documentation such as Figure~\ref{fig:system-overview}(b) may exceed an
agent's context window. We develop tools that let an agent navigate the
\method index efficiently (e.g.,  retrieving tables relevant to a query by RAG-top-k, or jumping to tables with known relationships to a target table by name). Our experiments
show that these capabilities are effective and necessary to scale to large lakes.

\underline{Budget-constrained optimization}: Despite our efficient, lightweight
indexing, processing every table in a large lake may still be costly. We
therefore formulate the \emph{Budget-Constrained Table Documentation} (BCTD)
problem, where given a fixed budget and a representative query log, we want to identify a subset of tables for processing that are most beneficial for the query workload. While we prove that the problem is computationally hard, we develop a fast density-based algorithm called  DensitySelect, which has provable approximation guarantees and is effective in experiments. 

\stitle{Contributions.} We make the following contributions:
\begin{itemize}[leftmargin=*,topsep=2pt,itemsep=1pt]
\item We propose \method, structured documentation for data lakes that
augments per-table descriptions with cross-table relationships such as confusable,
joinable, and union-able tables. We also provide agentic tools that make this
information easily accessible and navigable to LLMs (Section~\ref{sec:problem-setting}).
\item We develop lightweight, scalable approaches for extracting these
relationships, with a bounded linear construction cost despite the quadratic growth in possible table pairs
(Section~\ref{sec:method}).
\item We formulate the budget-constrained scenario as a principled optimization problem, 
and develop DensitySelect, a density-based algorithm with approximation guarantees
(Section~\ref{sec:budget-generation}).
\item We evaluate \method on diverse tasks such as data discovery, data analysis, and data science. Across 96 experimental settings (2 lakes x 6 backbone models x 8 benchmark tasks), \method outperforms diverse baselines including strong coding agents in 92 out of the 96 settings, all while operating at reduced cost/latency. 
(Section~\ref{sec:experiments}).
\end{itemize}


\section{Preliminaries and Problem Statement}
\label{sec:problem-setting}

\subsection{Preliminaries}
\paragraph{Data lake}
Let $\mathcal{L}=\{T_1,\ldots,T_n\}$ denote a collection of tables, each with a unique identifier, a schema and accessible rows.
For instance, in Example~\ref{ex:intro-sales},  we may encounter the many similarly-named tables (\code{Sales},
\code{Customers}, \code{Country}, etc.) in a multi-national's data lake, which can arise
from different regional
systems, historical sales records, and revenue summaries.

\paragraph{Data tasks on data lakes}
Users often need to perform diverse tasks on data lakes with thousands of tables, without knowing in advance
which tables in  $\mathcal{L}$ are relevant. We consider 3 such tasks in this work: data discovery, data analysis, and data science. 

 In \underline{data discovery}, let $Q$ be a user query with the intent to locate relevant tables in          $\mathcal{L}$ (e.g., ``\code{global sales by customer country since 2024}'' in Example~\ref{ex:intro-sales}), the goal is to retrieve the relevant set of tables that can be used to answer $Q$ (e.g., \{\code{Sales\_1}, \code{Sales\_14}, \code{Country\_1},
\code{Customer\_9}\} in Example~\ref{ex:intro-sales}).


In \underline{data analysis}, let $Q$ be a user's analysis query, and the goal is to answer $Q$ directly by executing Python or SQL code over relevant tables identified from $\mathcal{L}$. 
Today's reasoning-capable models can often interleave data discovery and analysis \cite{10.14778/3659437.3659452}, performing calculations as they explore tables (e.g., different \code{Sales} tables), and backtracking in a trial-and-error process, which can lead to very different trajectories compared to pure data discovery above.

In \underline{data science}, we consider the classical predictive modeling task of using a given target column $C$ in a table $T \in \mathcal{L}$ as the prediction target (which would be a regression task if $C$ is numerical,  or a classification task if $C$ is categorical). Agents performing such a predictive task may be able to leverage additional related tables in $\mathcal{L}$ (e.g., via joins or unions), which can boost predictive accuracy. 

\subsection{Problem Statement}
In this work, we study the benefit of providing agents with cross-table relationships in data lakes, by extracting 3 example types of common relationships: $\mathcal{R} = \{\text{Confusable, Joinable,}$
$\text{Union-able}\}$.
We study two problem variants, one without budget constraints and one with a fixed budget.

\begin{definition}
\label{def:no-budget}
\textbf{Relationship extraction in data lakes}. Given a data lake
$\mathcal{L}=\{T_1,\ldots,T_n\}$ and a set of target relationship types $\mathcal{R}$,
the goal is to extract relationships $\mathcal{R}$ on $\mathcal{L}$, denoted by:
\begin{align*}
&\mathcal{R}(\mathcal{L})=\{(T_i,T_j,r, d)\mid
\text{relationship }r\text{ holds for }T_i\text{ and }T_j\},\\[-2pt]
&\text{where }T_i,T_j\in\mathcal{L},\ i\ne j,r\in\mathcal{R}\ \text{and }d\text{ is its textual description}. \tag*{$\square$}
\end{align*}
\end{definition}

In our example with $\mathcal{R}=\{\text{Confusable},\text{Joinable},\text{Union-able}\}$, 
for each $T_i\!\in\!\mathcal{L}$, we use $\text{C}(T_i)$,
$\text{J}(T_i)$, and $\text{U}(T_i)$ to denote 
its Confusable, Joinable, and Union-able relationships, respectively.
The relationship $\mathcal{R}(\mathcal{L})$  we want to extract is therefore:
\[
\mathcal{R}(\mathcal{L})=\bigcup_{T_i\in\mathcal{L}}\bigl(
\text{C}(T_i)\cup\text{J}(T_i)
\cup\text{U}(T_i)\bigr).
\]

 Figure~\ref{fig:system-overview}(b) shows an example of such extracted relationships, presented in a markdown table format.
Given the pairwise nature of cross-table relationships, a key challenge is to produce such relationships accurately and efficiently (without quadratic comparisons), which we discuss in Section~\ref{sec:method}.

Although our approach solves the problem in Definition~\ref{def:no-budget} efficiently, at a cost as low as one cent per table using LLMs, extracting relationships for every table in large data lakes containing millions of tables may be neither necessary, nor desirable for users with a fixed budget. We therefore also study a budget-constrained variant of the problem, formalized in Definition~\ref{def:with-budget} below, which assumes access to a historical query log $QL$ that is representative of workloads on $\mathcal{L}$.


\begin{definition}
\label{def:with-budget}
\textbf{Relationship extraction in data lakes with a budget constraint}.  Given  a data lake $\mathcal{L}=\{T_1,\ldots,T_n\}$, a set of target cross-table relationships: $\mathcal{R}$, and extractable relationships $\mathcal{R}(\mathcal{L})$ as in Definition~\ref{def:no-budget}. Using a fixed budget $B$ that can be used to extract $\mathcal{R}(\mathcal{L})$, and a representative query log $QL$, select a subset $S \subseteq \mathcal{R}(\mathcal{L})$ expected to be the most beneficial for the given $QL$.
\end{definition}

While we state this budget-constrained version at a high level, we study a concrete combinatorial optimization problem corresponding to this high level definition in Section~\ref{sec:budget-generation}, for which we develop approximation algorithms with provable guarantees.

\section{\method}
\label{sec:method}

\subsection{System Overview}
At a high level, \method aims to build a relationship graph in large data lakes offline, enriched with natural-language context describing how tables differ or relate, which we posit will help agents to better navigate data lakes and perform complex tasks. Figure~\ref{fig:system-overview} gives an overview of our proposed \method system. 

In the offline indexing step shown in Figure~\ref{fig:system-overview}(a), we  extract three example cross-table relationships, $\mathcal{R} = \{\text{Confusable, Joinable,}$ $ \text{Union-able}\}$ (which are extensible to additional types of relationships), using lightweight processing that scales linearly with the number of tables if there are no budget constraints (Section~\ref{sec:offline}), or sub-linearly with a fixed budget constraint  (Section~\ref{sec:budget-generation}). In either case, the unit cost of our proposed processing is lightweight, costing as low as \$0.01 per table even in large data lakes, when open-source models are used.

Figure~\ref{fig:system-overview}(b) shows an example output, represented in a structured manner for navigation. Each row here represents a table, with basic table/column descriptions as well as statistical profiles shown in the second column ``\code{Basic info (B)}'', that can be produced by prompting LLMs similar to prior work~\cite{autoddg}. In addition, we represent the three example cross-table relationships $\mathcal{R} = \{\text{Confusable, Joinable, Union-able}\}$  as three additional columns to the right, with rich grounded context for why certain tables either differ or relate, 
which are made possible by today's LLMs.

In addition to cross-table relationships, we also study \emph{table-to-code relationships}, like illustrated in the last column of Figure~\ref{fig:system-overview}(b). 
\iftoggle{full}
{
}
{
    In the interest of space, we defer discussions of this relationship and its experiments to a full technical report~\cite{full}. 
}

At online query time in Figure~\ref{fig:system-overview}(c), given a user task on the data lake --- which can be data discovery, data analysis or data science --- we show that suitably-designed agentic \method tools can enable LLMs to navigate a large collection of tables, by jumping from an initial candidate table, to additional tables that may be confusable/joinable/union-able as needed using the index and the tools we build, 
which we show is both much more effective (with higher task success rates), and efficient (with lower token usage and latency). We  will explain the online step in Section~\ref{sec:online}.

\subsection{Offline Index: Extract Relationships}
\label{sec:offline-step}

We now describe how we instantiate \method by extracting three example table relationships, $\mathcal{R} = \{\text{Confusable}$, $\text{Joinable,}$  $\text{Union-able}\}$ in turn below,  to populate the corresponding fields labeled as $C(T)$/$J(T)$/$U(T)$, for each $T \in \mathcal{L}$ in Figure~\ref{fig:system-overview}(b). These are  on top of basic per-table descriptions in ``\code{Basic info (B)}'',  denoted by $B(T)$, that contain table/column descriptions, primary-keys, grains, and statistical profiles that we obtain by prompting LLM once per-table using implementations similar to prior work~\cite{autoddg, taco}.

\stitle{\underline{Confusable Tables}.} 
\label{sec:offline}
In a large data lake, it is not uncommon to have many similar-looking tables (e.g., 17 ``\code{Sales}'' tables like discussed in Example~\ref{ex:intro-sales}). Intuitively, we want our offline index to help LLMs to understand how these tables differ, and which ones are really relevant to a task at hand, without needing to open up each and every \code{Sales} table to inspect at query time. This type of relationship has received little attention in the literature, but is now possible to study with modern LLMs.

The key challenge here lies in efficiently finding pairs of possibly confusable tables in a large lake $\mathcal{L}$. While files with the identical name are clearly confusable, it is unclear if other  similarly-named tables are also confusable with a given target table (e.g., ``\code{Sales\_latest}'' with ``\code{Sales}''). Given the size of $\mathcal{L}$, we need to identify such potentially confusable tables efficiently without quadratic comparisons.

\begin{algorithm}[h]
\small
\caption{Query-driven discovery of confusable tables
}
\label{alg:confusable-discovery}
\Input{Tables $\mathcal{L}$, basic description $\{B(T)\}_{T\in\mathcal{L}}$
}
\Output{Symmetric confusable group pairs $E$ 
}
$E\leftarrow\emptyset$\;
\ForEach{$T\in\mathcal{L}$}{
  $Q_T\leftarrow$ Generate $k$ LLM-brainstormed questions about $T$\nllabel{line:confusable-probes}\;
  \ForEach{$q\in Q_T$}{
    $S_q\leftarrow$ top $N$ tables retrieved for $q$ by embedding in $\{B(T)\}_{T\in\mathcal{L}}$\nllabel{line:confusable-retrieve}\;
    $R_q\leftarrow$ LLM ranked list of tables in $S_q$ for $q$\; \nllabel{line:confusable-rank}
    \ForEach{$T'$ ranked higher than $T$ in $R_q$\nllabel{line:confusable-candidates}}{
    $E\leftarrow E\cup\{(T, T'),(T', T)\}$
    }
  }
}
\Return{$E$}\; \nllabel{line:return}
\end{algorithm}

We develop a \emph{query-driven} approach to identify such confusable tables, like shown in Algorithm~\ref{alg:confusable-discovery}. Specifically, for each target table $T$ that we want to find confusable tables for, we ask LLM to brainstorm a small number of $k$ (e.g., $k=3$) questions that this $T$ can be used to answer (Line~\ref{line:confusable-probes}). For each such question $q$, we compare $q$'s embedding with the embedding of the basic descriptions of all tables $\{B(T)\}_{T\in\mathcal{L}}$, and retrieve the top-N (e.g., $N=20$) candidates denoted as $S_q$ (Line~\ref{line:confusable-retrieve}). We then ask an LLM to rank $S_q$ using a ranked list $R_q$, ordered from the most-likely to least-likely related to $q$, so that the LLM will likely open and inspect files in that order (Line~\ref{line:confusable-rank}). Since we know $T$ can be used to answer $q$ (as $q$ is a question brainstormed from $T$), we treat all tables that rank higher in $R_q$ than $T$ as potentially confusable with $T$. We repeat the process and return the set of confusable table pairs as $E$ (Line~\ref{line:return}). 

Such pairs can then be sent to LLMs, together with their basic descriptions $B(T)$, as well as profiling information (e.g., grains, primary-keys, value-ranges and statistical profiles for each column) collected while generating $B(T)$, to generate a description that can best distinguish this pair of confusable tables, like illustrated in the column ``\code{Confusable tables (C)}'' of Figure~\ref{fig:system-overview}.

\begin{example}
\label{ex:confusable-discovery}
We continue with Example~\ref{ex:intro-sales} to show how confusable tables are generated. For each table $T$, we first ask LLM to brainstorm a small number of $k$ questions, $Q_T$, that can be answered by $T$. Suppose for the first table ``\code{Sales\_1}'' in Figure~\ref{fig:system-overview}, LLM generates a query $q \in Q_T$ as ``\code{store sales by country}''. We embed $q$ and compare its embedding with the embeddings of basic description of all tables, $\{B(T)\}_{T\in\mathcal{L}}$, to retrieve the top-N candidate tables with the highest embedding similarity to $q$, denoted by $S_q$ (which may include all 17 ``\code{sales}'' tables as well as others like ``\code{sales\_latest}''). 

We then ask the LLM to rank $S_q$ based on their table descriptions and a sample of rows for each table in $S_q$. Suppose the LLM-ranked list is $R_q =$ \{\code{sales\_5}, \code{sales\_17}, \code{sales\_1}, \ldots \}. We find that \code{sales\_5} and \code{sales\_17}  are likely confusable with \code{sales\_1} for the LLM, since they rank higher than \code{sales\_1}. 
To distinguish them, we ask the LLM to explain their differences. Perhaps \code{sales\_1} is for ``global sales'' while \code{sales\_5} is for ``US sales'', and \code{sales\_1} is for ``retail store sales'' while \code{sales\_17} is for ``stock sales'', etc., like shown in Figure~\ref{fig:system-overview}(b).  
\end{example}

\paragraph{\underline{Cost analysis}}
The per-table cost of annotating confusable tables for each target table $T$, denoted by $c(C)$, is at most $c(C) = (1 + kN)$ LLM calls, where both $k$ and $N$ are small constants. Here, 1 is for generating $k$ brain-stormed questions for $T$, and $kN$ is to explain all tables ranked higher than $T$ (which is at most $N$), for all $k$ questions.
Note that in practice the number of LLM calls is much smaller, as LLM is often highly capable and can rank the target $T$ high in its ranked list -- empirically, we find that on average LLM ranks only $1.7$ tables above the target $T$ (or $N = 1.7$ empirically). Regardless,  $c(C)$  scales only linearly in the size of the data lake $\mathcal{L}$ as both $k$ and $N$ are fixed constants.


\begin{figure*}[t]
\vspace{-18mm}
  \centering
  \begin{minipage}[t]{0.635\textwidth}
    \centering
    \includegraphics[height=44mm]{figures/join-example.png}
    \par\vspace{-2mm}
    \caption{An example for extracting joinable relationships.}
    \label{fig:join-example}
  \end{minipage}\hfill
  \begin{minipage}[t]{0.365\textwidth}
    \centering
    \includegraphics[height=42mm]{figures/union2.png}
    \par\vspace{-2mm}
    \caption{Union-able tables in data lakes.}
    \label{fig:union}
  \end{minipage}
  \par\vspace{-2mm}
\end{figure*}

\stitle{\underline{Joinable Tables}}
\label{sec:method-j}
While join detection is a classical topic in data management as we will review in related work (Section~\ref{sec:related}), detecting joinable tables efficiently and accurately in large data lakes with tens of thousands of tables is still challenging.

Leveraging recent advances in both LLM and embeddings, we develop a scalable and accurate approach for joinable relationships in large data lakes, that adapts key ingredients from classical join detection~\cite{zhang2010foreign, papenbrock2015binder, chen2014fast, rostin2009machine}, as outlined in Algorithm~\ref{alg:joinable-discovery}.

\begin{algorithm}[h]
\small
\caption{Extract joinable relationships}
\label{alg:joinable-discovery}
\Input{Tables in a lake $\mathcal{L}$}
\Output{Candidate joins $\mathcal{J}$}
$\mathcal{P}\leftarrow\emptyset$; $\mathcal{F}\leftarrow\emptyset$; $\mathcal{J}\leftarrow\emptyset$\;
\ForEach{$T\in\mathcal{L}$}{
  $(P_T,F_T)\leftarrow$ LLM-inferred PK/FK on $T$ using $B(T)$\nllabel{line:joinable-pkfk}\;
  $\mathcal{P}\leftarrow\mathcal{P}\cup P_T$; $\mathcal{F}\leftarrow\mathcal{F}\cup F_T$\;
}
$Emb(\mathcal{P})\leftarrow$ Embed $\mathcal{P}$ by (column names + sample values) \nllabel{line:joinable-embed}\;
\ForEach{$f\in\mathcal{F}$}{

    $Emb(f) \leftarrow$ Embed $f$ by (column names + sample vals) \nllabel{line:joinable-embed-fk}\;
    
    $C_f\leftarrow$ Retrieve top-N PK in $Emb(\mathcal{P})$ using $Emb(f)$, that pass PK/FK containment checks \nllabel{line:joinable-retrieve-topk}\;
    
  
    $P_f\leftarrow$ LLM validates joined-sample, ranks semantically plausible top-k\nllabel{line:joinable-validate}\;
    $\mathcal{J}\leftarrow\mathcal{J}\cup\{(f,p):p\in P_f\}$\nllabel{line:joinable-emit}\;

}
\Return{$\mathcal{J}$}\nllabel{line:joinable-return}\;
\end{algorithm}

\underline{[Key identification and enrichment]}. In Line~\ref{line:joinable-pkfk} of Algorithm~\ref{alg:joinable-discovery}, we start by using LLM to identify candidate primary-keys (PK) and foreign-keys (FK) in each table, using a single LLM call. Specifically, we provide LLM with each table's schema and sample rows, together with statistical profiles in $B(T)$ (uniqueness, distinct values, etc.). Using statistical and semantic information, we ask LLMs not only to identify candidate PK/FK in each table (which they can predict very accurately), but also to enrich each PK/FK with alternative column names that are semantically meaningful.

Figure~\ref{fig:join-example} shows such an example, where columns with blue boxes indicate LLM-predicted PK candidates, and columns with orange boxes indicate LLM-predicted FK candidates. LLM also enrich these PK/FK column names with alternative names, such as enriching the PK  ``\code{id}'' column in ``\code{Sales\_1}'' on the left, as ``\code{id}'' $\rightarrow$ ``\code{sales\_id}'', ``\code{transaction\_id}'' etc., as well as enriching the PK candidate ``\code{c\_id}'' in the two tables on the right, as ``\code{c\_id}'' $\rightarrow$ ``\code{country\_id}'', etc.

\underline{[Embed PKs]}. In Line~\ref{line:joinable-embed} of Algorithm~\ref{alg:joinable-discovery}, we then embed each candidate PK column, using their original and alternative column-names, together with column-values, to produce $Emb(\mathcal{P})$. 
This is shown in the example of Figure~\ref{fig:join-example} with the step marked as \textcircled{1}, where the PK column ``\code{c\_id}'' for table ``\code{Country\_5}'' is embedded using its original and enriched column names (``\code{c\_id}'', ``\code{country\_id}''), as well as sample-values from the column (``\code{1}'', ``\code{2}'', etc.). 

\underline{[Embed FK for top-K]}.
Line~\ref{line:joinable-embed-fk} embeds each candidate FK column $f$ using the same column-names with sample-values to produce $Emb(f)$; in Line~\ref{line:joinable-retrieve-topk}, we use $Emb(f)$ to retrieve top-N  $Emb(\mathcal{P})$, i.e., the PK columns with the highest similarity to the target $f$. The top-N candidates are checked for PK/FK value containment, with passing pairs retained.

Using the example in Figure~\ref{fig:join-example}, with the ``\code{FK \#2}'' in table ``\code{Sales\_1}'' as the target FK $f$ like shown in step \textcircled{2}, we would embed its column names and values, to retrieve its embedding top-K PK-candidates like shown in step \textcircled{3}, which include column ``\code{c\_id}'' from table ``\code{Country\_5}'' and  ``\code{Country\_9}''  shown on the right (note that these ``\code{c\_id}'' columns are ranked high as embedding top-K for  $Emb(f)$, because both are enriched with  meaningful col-names like ``\code{country\_id}'', thus producing good matches with $Emb(f)$). Both of these two ``\code{c\_id}'' columns would be checked for PK/FK value containment checks and pass such checks. 

\underline{[Semantic validation]}.
Finally, in Line~\ref{line:joinable-validate}, we perform an actual join using candidate PK/FK columns, and ask LLM to check samples from the joined row-pairs to validate semantic plausibility. For instance, in Figure~\ref{fig:join-example}, even though both ``\code{country\_id}'' columns from ``\code{Country\_5}'' and  ``\code{Country\_9}'' are retrieved as top-N, an LLM-validation step of the joined row-pairs would reveal the join with ``\code{Country\_5}'' as invalid, because joining the country-names from ``\code{Country\_5}'' into ``\code{Sales\_1}'' shows semantic incompatibility (a US address like ``\code{Springfield, IL}'' in country ``\code{afghanistan}''), which is thus pruned away; while the join with ``\code{Country\_9}'' is semantically plausible and retained as a final PK candidate for ``\code{FK \#2}'' (``\code{country\_id}'') in the figure.  All validated join relationships are then documented in \method like shown in column ``\code{Joinable tables (J)}'' of Figure~\ref{fig:system-overview}.

Note that for each FK candidate, Algorithm~\ref{alg:joinable-discovery} may return more than one PK candidate after all these steps, as in large data lakes PK/FK candidates may not be uniquely resolved. We do not force a unique PK/FK pairing in such cases, and use ``late-binding'' that retains all plausible pairs in \method, and defer this to query-time for LLM to resolve (where a user query may contain additional disambiguating information, such as referencing a column that is only present in one of the candidate pairs).

\paragraph{\underline{Cost analysis}} The per-table cost of extracting join-able relationships  for each table $T \in \mathcal{L}$ using Algorithm~\ref{alg:joinable-discovery}, denoted by $c(J)$, is at most $c(J) = (1 + |F_T|)$ LLM calls, where 1 is for the LLM call to identify PK/FK candidates in $T$ (Line~\ref{line:joinable-pkfk}), and $|F_T|$ indicates the number of FK candidates on a table $T$, where the $|F_T|$ LLM calls are for validating and ranking the semantic-plausibility for each of the $|F_T|$  PK/FK candidates (Line~\ref{line:joinable-validate}). Since $|F_T|$ is the number of candidate FKs per table, it is usually a very small constant (empirically, less than 2 for data lakes we tested). The cost $c(J)$ is therefore a small constant, making its cost linear in the number of tables in $\mathcal{L}$.


\stitle{\underline{Union-able Tables}}. \label{sec:method-u}
Union-able tables are common in data lakes, and are often partitioned by time or geographic boundaries, because data pipelines would often dump data at regular intervals (e.g., daily or hourly), thereby creating sub-partitioned tables with identical schemas and similar filenames following a fixed pattern, like the examples from real data lakes in Figure~\ref{fig:union} would show. This is widely reported in prior studies~\cite{union-partition, union-partition-1, union-partition-2} and is very common in real data lakes, as such union-able tables are typically ``machine-generated'' by recurring pipelines~\cite{union-partition-1, union-partition-2}.

For this reason, we use a reference implementation to extract union-able tables, by clustering tables based on schema-names, and then present LLM with all tables in each cluster $C$ (with sample rows and statistical profiles in $B(T)$ for each table $T \in C$), so that LLM can determine which subsets are union-able and along which attribute, which are documented like shown in column  ``\code{Union-able tables (U)}'' of Figure~\ref{fig:system-overview}.

Note that while our reference implementation for union-able tables focuses on schema-aligned table partitions, other existing union-search methods~\cite{nargesian2018table, fan2023starmie, khatiwada2023santos, bogatu2020d3l, hu2023autotus}
can be readily plugged into our framework through the same interface to enhance its coverage.

\paragraph{\underline{Cost analysis}}
Because each table $T\in\mathcal{L}$ can only fall into at most 1 cluster $C$,
each table will be scanned at most once by the LLM for union-able relationships. The per-table cost for union-able relationships is therefore $c(U) = 1$, which also grows linearly as the lake $\mathcal{L}$ grows.

\subsection{Online: Agentic Tools for \method}
\label{sec:online}

In a large data lake $\mathcal{L}$ with thousands or millions of tables, even succinct \method documentation like in Figure~\ref{fig:system-overview} can be large (e.g., it can have millions of entries with one for each table), which would be expensive to load entirely into an LLM agent's context window indiscriminately, as it degrades both the agent's latency and quality (e.g., lost in the middle~\cite{liu2024lost}). 

We therefore develop a set of agentic tools~\cite{qu2025tool, schick2023toolformer} to work with \method, like illustrated in Figure~\ref{fig:system-overview}(c). Today's LLMs are capable of invoking these tools appropriately, to selectively load only a relevant subset of \method, and follow observed relationship-pointers to navigate across related tables on demand. We describe a few core tools below in turn:


\paragraph{\underline{\texttt{Get\_readme\_by\_RAG\_top\_k(query, start\_pos, k)}}}
This ``RAG-top-k'' tool takes a user input query, and compares its embedding with the row-wise embedding of all table entries in \method to retrieve the top-k most relevant to the user query. By default we set \code{start\_pos} = 0, and $k = 50$ for agents to retrieve the first-50 matches, but the agent can always inspect the first-50 and if not finding relevant matches, retrieve the next $k$ by changing the \code{start\_pos} parameter and invoke this tool again, which would retrieve the next batch of matches, etc.

\paragraph{\underline{\texttt{Get\_readme\_by\_table\_names([table\_list])}}} This ``Retrieve-by-table-name'' tool allows an agent to specifically retrieve a list of readme entries by name. It is useful because an agent may inspect an initial set of matches, find a relevant starting table $T$, and then decide that there is a need to jump to another relevant table $T'$ that may be confusable/joinable/union-able with $T$ for additional inspection (e.g., to see what $T'$ actually contains, and what other tables $T'$ can connect to). Invoking this tool allows LLMs to jump directly to these related tables by name. 


\paragraph{\underline{\texttt{Explore\_table(table)}}}
This tool lets the agent actually open a table and explore data in a table, which agents can use for final verification or to resolve fine-grained details not present in \method. 

We also allow agents to use general python/sql tools, in addition to these readme-specific tools described above.

\begin{example}
\label{ex:tool}
We revisit Example~\ref{ex:intro-sales} in Figure~\ref{fig:system-overview} to show how an agent may use these tools to navigate \method. Suppose the input query is
\textit{``global retail sales by customer country since 2024''},  the agent first invokes \code{Get\_readme\_by\_RAG\_top\_k} to retrieve
the top-50 tables in readme that are the most similar to the query, which include many (but not all)
\code{Sales}, \code{Customer} and  \code{Country} tables in the lake. The agent inspects all returned tables' readme entries, identifying \code{Sales\_1} as a possible candidate after reading fields such as ``basic info (B)'' and ``confusable tables (C)'' in the example. 

From \code{Sales\_1}, it may read its ``joinable tables (J)'' and ``union-able tables (U)'' fields and decide to navigate these connections, to \code{Sales\_14} (union-able with \code{Sales\_1} to answer the query), as well as  \code{Country\_1} and  \code{Customer\_9} (to join in order to answer the question). These additional tables may not be in the initial batch retrieved using RAG-top-k, in which case the agent can invoke \code{Get\_readme\_by\_table\_names()} to directly retrieve the readme for these tables using names.

Note that the ``joinable tables (J)'' field may contain more than 1 plausible join for an FK column, if semantic-validation (Line~\ref{line:joinable-validate} of Algorithm~\ref{alg:joinable-discovery}) cannot uniquely identify one PK column to join with. In such cases readme has more than one join peers, which LLM can open and explore, to perform what we call a query-time ``late-binding'' to decide which plausible tables to actually bring in to answer user query.
\end{example}

\section{\method: Budget-Constrained Version}
\label{sec:budget-generation}
So far we studied how to extract relationships for all tables in a lake.
We now consider the budget-constrained version in Definition~\ref{def:with-budget}.

\subsection{Problem Formulation: BCTD}
\label{sec:budget-coverage-model}
While our system described above can already process tables efficiently and in a cost-effective manner (costing as low as 1 cent per table), on really large lakes with millions of tables, and/or when a user has a fixed monetary budget $B$, we can instead focus on a subset of tables that is expected to be the most beneficial.

Because users are unlikely to hand-pick or manually specify what tables are useful to index, we rely on a ``query log'' $QL$, that can be a set of historical jobs/scripts/sql-queries representative of future queries, which are often available~\cite{log1,log2} and are leveraged in prior studies~\cite{yan2020auto, yang2009recommending}. 

Intuitively, from the $QL$, we can extract what tables are being used and at what frequency, which is a proxy for how important each table $T$ is, and how beneficial it may be to index $T$. 
Let $QL=(Q_1,\ldots,Q_m)$ be the query log, we use 
$\mathcal{T}_i\subseteq\mathcal{L}$ to denote the set of tables used by each $Q_i$ (e.g., all tables accessed by this sql query or python script $Q_i$). 

Given a monetary budget $B$ for indexing, and since \method is indexing tables at the table granularity (Figure~\ref{fig:system-overview}), we reduce this to a table selection problem -- namely, we want to select a subset of tables  $D\subseteq\mathcal{L}$ for indexing, so as to maximize its possible benefit on $QL$. 

On the cost side, for a selected table set $D$, since each table populates the same  fields and follows the same steps in Section~\ref{sec:offline}, the offline cost of processing each table in $D$ is similar and can be pre-determined as a constant $\bar{\kappa}$ (e.g., based on the average number of tokens consumed per table, and LLM token price).  For a fixed budget $B$, this means that the number of tables we can select in $D$ cannot exceed $\lfloor B/\bar{\kappa}\rfloor$, or $|D| \leq \lfloor B/\bar{\kappa}\rfloor$.

On the utility side, for a selected table set $D$ that is indexed, the resulting readme is reliably beneficial for a query $Q_i \in QL$, only if all tables accessed in $Q_i$, written as $\mathcal{T}_i$, are contained in $D$ (as otherwise if a related table is missing from the readme, agents will likely have to fall back to brute-force exploration). We can therefore write the utility function     as 
  $F(D) = \sum_{i=1}^{m}\mathbf{1}\{\mathcal{T}_i\subseteq D\}$,
where $\mathbf{1}\{\cdot\}$ is the indicator function, and $F(D)$ indicates how many queries in the query log $QL$ can be ``covered'' by our selected tables $D \in \mathcal{L}$.

We are now ready to reduce the high-level problem stated in Definition~\ref{def:with-budget}, to an exact optimization problem that we term Budget-Constrained Table Documentation, defined below:

\begin{definition}
\label{def:cctd}
\textbf{Budget-Constrained Table Documentation (BCTD)}.
Given a data lake $\mathcal{L}$, a query log $QL = \{Q_1, \ldots, Q_m\}$, the set of tables used by each of the queries $\{\mathcal{T}_1, \ldots, \mathcal{T}_m\}$, a total budget $B$ and a per-table processing cost $\bar{\kappa}$, select a subset of tables $D \subseteq \mathcal{L}$, such that: 
\begin{align}
\max_{D \subseteq \mathcal{L}} \quad
& \sum_{i=1}^{m} \mathbf{1}\{\mathcal{T}_i \subseteq D\}
\label{eq:objective} \\
\text{s.t.} \quad
& |D| \leq \lfloor B/\bar{\kappa} \rfloor
\notag
\end{align}
\end{definition}


We first establish that BCTD is equivalent to \emph{Densest
$k$-Subhypergraph} (DkSH)~\cite{chlamtac2018densest}. Using this
equivalence, we develop a selection algorithm in the DkSH formulation
and map its solution back to BCTD.

\begin{lemm}[Equivalence to DkSH and hardness]
\label{lem:bctd-dksh-equivalence}
BCTD is polynomially equivalent to DkSH with repeated hyperedges
counted by multiplicity. The reductions preserve the budget and objective, therefore BCTD is NP-hard.
\end{lemm}

\iftoggle{full}
{
    \begin{proof}
    To reduce BCTD to DkSH, create vertices
    $V=\{v_T:T\in\mathcal{L}\}$ and one hyperedge
    $e_i=\{v_T:T\in\mathcal{T}_i\}$ per logged query $Q_i$.
    For a table set $D$, let $X_D=\{v_T:T\in D\}$. Then
    \[
    \begin{aligned}
      |X_D|&=|D|, \quad
      F_H(X_D)&=\sum_{i=1}^{m}\mathbf{1}\{e_i\subseteq X_D\}=F(D),
    \end{aligned}
    \]
    where $F_H$ counts contained hyperedges, including repetitions.
    
    Conversely, reduce DkSH to BCTD by creating one table per vertex
    and one logged query per hyperedge occurrence. Every vertex set $X$
    maps to $D_X=\{T:v_T\in X\}$ with $|D_X|=|X|$ and
    $F(D_X)=F_H(X)$. Both constructions are polynomial in the explicit
    query/hyperedge incidences and preserve the budget and objective.
    We use the at-most-$k$ form, which has the same optimum as selecting
    exactly $\min\{k,|V|\}$ vertices because $F_H$ is monotone.

    Since ordinary DkSH is a special case and is
    NP-hard~\cite{arulselvan2014setunion}, BCTD is NP-hard.
    \end{proof}    
}
{
We give a full proof of Lemma~\ref{lem:bctd-dksh-equivalence} in our technical report~\cite{full}.
}

\subsection{DensitySelect for BCTD}
\label{sec:budget-algorithm}
Based on Lemma~\ref{lem:bctd-dksh-equivalence}, we use DkSH to develop an algorithm for BCTD. We first merge identical hyperedges and assign each its query multiplicity as a weight without changing $F_H$.
We identify vertices with their tables and write $F$ for the common
objective.

\paragraph{\underline{Selection algorithm}}
Because hyperedges can share vertices, selecting multiple hyperedges costs the
size of their vertex union, not the sum of their individual sizes.
We conservatively approximate this union cost with additive costs
to obtain DensitySelect, an efficient density-based selection
algorithm for BCTD.

We first discard queries touching more than $k$ tables and set aside
empty-table queries, whose constant coverage is $w_0$. Merge identical
hyperedges from the remaining queries into distinct nonempty hyperedges
$A_1,\ldots,A_M$, with weight $w_j$ equal to the number of queries
represented by $A_j$. If $M=0$, the empty selection is optimal.
Otherwise, define
\[
  d_T=|\{j:T\in A_j\}|,\qquad
  \Delta=\max_T d_T,\qquad \ell=\max_j|A_j|.
\]
Here $\Delta$ is the maximum number of candidate hyperedges incident
to one vertex, and $\ell$ is the maximum hyperedge size.

In DensitySelect shown in Algorithm~\ref{alg:budget-fast}, we assign each hyperedge the additive cost $c_j=|A_j|$.
We then sort hyperedges by nonincreasing
$w_j/c_j$, admitting a hyperedge when the sum of admitted hyperedge costs
remains at most $k$, and skipping it otherwise. It compares the resulting
vertex union with the vertex set of the highest-weight individual hyperedge,
returning whichever covers more queries. Density and single-hyperedge
weight ties use hyperedge index; the packed selection wins a coverage tie.

We show that DensitySelect  in Algorithm~\ref{alg:budget-fast} has the following approximation guarantee.

\begin{algorithm}[t]
\small
\caption{DensitySelect: BCTD density-based selection}
\label{alg:budget-fast}
\Input{Candidate hyperedges $(A_j,w_j)_{j=1}^{M}$; vertex budget $k$}
\Output{Selected vertices $D_{\mathrm{alg}}$}
\If{$M=0$}{\Return{$\emptyset$}\;}
Set $c_j\leftarrow |A_j|$ and order hyperedges by nonincreasing $w_j/c_j$\;
$D\leftarrow\emptyset$; $b\leftarrow 0$\;
\ForEach{$j$ in this order}{
  \If{$b+c_j\le k$}{
    $D\leftarrow D\cup A_j$; $b\leftarrow b+c_j$\;
  }
}
$j^\star\leftarrow\argmax_j w_j$\;
\Return{$\argmax_{X\in\{D,A_{j^\star}\}}F(X)$}\;
\end{algorithm}

\begin{lemm}[Approximation and running time]
\label{lem:cctd-approximation}
For the weighted hypergraph above, let $F(D)$ denote the total weight
of hyperedges contained in a vertex set $D$, including the constant
weight $w_0$ of empty hyperedges.
Let $\mathrm{OPT}=\max_{|D|\le k}F(D)$. For $M>0$,
DensitySelect returns a set $D_{\mathrm{alg}}$ with
$|D_{\mathrm{alg}}|\le k$ and
  $F(D_{\mathrm{alg}})\ge\frac{1}{2\Delta}\,\mathrm{OPT}$.
Given the hyperedges, its running time is
$O(M\log M+M\ell)$.
\end{lemm}

\begin{proof}
The packed union contains at most $b\le k$ vertices, and every candidate
hyperedge has at most $k$ vertices. Thus either returned candidate is feasible.

Let $D^\star$ be an optimal vertex set and let
$J^\star=\{j:A_j\subseteq D^\star\}$ index its covered hyperedges.
Their additive cost satisfies
  $\sum_{j\in J^\star}c_j\le\Delta|D^\star|\le\Delta k$.
Taking a $1/\Delta$ fraction of each such hyperedge is therefore feasible
for the fractional relaxation of the additive-cost problem over nonempty
hyperedges with capacity $k$. If $V_{\mathrm{frac}}$ denotes its optimum, then
$V_{\mathrm{frac}}\ge(\mathrm{OPT}-w_0)/\Delta$.

If no hyperedge is rejected, every candidate hyperedge is contained
in the packed union, so the selection is optimal. Otherwise, let
$r$ be the first rejected hyperedge and $P$ the total weight admitted
before it. In density order, the fractional solution consists of this
prefix and at most a fraction of
$r$, so $V_{\mathrm{frac}}\le P+w_r$. The completed packing has value
at least $w_0+P$, while the highest-weight single-hyperedge candidate
has value at least $w_0+w_r$. The better candidate therefore satisfies
\[
\begin{aligned}
  F(D_{\mathrm{alg}})&\ge w_0+\tfrac12(P+w_r)\quad
  &\ge w_0+\frac{\mathrm{OPT}-w_0}{2\Delta} \quad
  &\ge\frac{\mathrm{OPT}}{2\Delta}.
\end{aligned}
\]

\noindent
Sorting takes $O(M\log M)$ time, while building the packed union and scoring
the candidates takes $O(M\ell)$, thus showing the $O(M\log M+M\ell)$ running time.
\end{proof}

Finally, we map the returned vertices to their corresponding tables.
The equivalence preserves the table budget and query log coverage,
so for $M>0$, Lemma~\ref{lem:cctd-approximation} yields a feasible
BCTD solution with at least a $1/(2\Delta)$ fraction of optimal
coverage. Since $M\le m$, its running time is also
$O(m\log m+m\ell)$.
Note that this algorithm uses only query-table incidence data and makes no LLM calls. 

\section{Experiments}
\label{sec:experiments}

We perform extensive evaluation to understand the effectiveness of \method on diverse tasks. 

\label{sec:experiment-setup}

\paragraph{\textbf{Data lakes}}
We create two experimental data lakes by pooling tables from public
benchmarks BIRD~\cite{bird}, Spider~\cite{spider}, KramaBench~\cite{kramabench}
and BI-Bench~\cite{bibench}:

\underline{\emph{\lakeoriginal{}}}: We pool all tables together from these 4 sources, and deduplicate at the table-level (as some public data sources are used by multiple benchmarks),
which results in 4,127 tables. We refer to this as \lakeoriginal{}.

\underline{\emph{\lakeaugmented{}}}: To create an even larger and more challenging lake, we use LLM to synthesize  group-by queries on randomly selected tables (to simulate user behaviors of performing aggregate analysis and then saving results~\cite{diamantini2026analytic}) to create 3,220 derived tables, and horizontally partitioned tables (to simulate partitioned-files that are common in data lakes~\cite{nargesian2018table, union-partition}, e.g. in Figure~\ref{fig:union}) to create 5,482 tables.
The resulting \emph{\lakeaugmented{}} contains 12,829 tables.

Table~\ref{tab:lake-statistics-cost-preview} summarizes the statistics of these lakes, and the offline construction costs to build \method on both lakes using all 6 backbone models tested. It can be seen that using open-source models like GPT-OSS-120B or DeepSeek-V4-Flash, the unit cost to process each table is as low as $\$0.01$, making \method economically practical without any budget constraints. 

\begin{table*}[tp]
\par\vspace{-19mm}
  \caption{Data lake statistics and construction costs.}
  \label{tab:lake-statistics-cost-preview}
  \vspace{-2mm} 
  
\savebox{\tblbox}{%
\small
\setlength{\tabcolsep}{4pt}
\begin{tabular}{@{}lrrrr@{\hspace{6pt}}|@{\hspace{6pt}}rrrrrr@{}}
\toprule
\multirow{2}{*}{Lake} & \multicolumn{4}{c@{\hspace{6pt}}|@{\hspace{6pt}}}{Statistics} & \multicolumn{6}{c}{Construction cost (USD, per-table / total)} \\
\cmidrule(lr){2-5}\cmidrule(lr){6-11}
 & Tables & Avg. rows & Avg. columns & Avg. cells & GPT-5.4 & GPT-5.4-mini & GPT-5.4-nano & GPT-OSS-120B & DeepSeek-V4-Pro & DeepSeek-V4-Flash \\
\midrule
\lakeoriginal{} & 4,127 & 97,359 & 14.71 & 554,125 & 0.299 / 1,235 & 0.090 / 371 & 0.024 / 101 & 0.017 / 71 & 0.061 / 253 & 0.015 / 63 \\
\lakeaugmented{} & 12,829 & 47,490 & 10.99 & 294,848 & 0.222 / 2,854 & 0.067 / 857 & 0.018 / 233 & 0.013 / 162 & 0.045 / 580 & 0.011 / 145 \\
\bottomrule
\end{tabular}}
\usebox{\tblbox}

  \par\vspace{-2mm} 
\end{table*}

\paragraph{\textbf{Tasks}.} We evaluate 3 diverse data tasks using these data lakes:

\underline{Data discovery}: In this task, an agent needs to identify  tables in a lake relevant for a user query. We use  queries in (BIRD, Spider, KramaBench, BI-Bench) as input, and measure \emph{exact-set match} 
between the top-$K$
predicted tables and the ground-truth tables, where $K$ is the size of the ground-truth.

\underline{Data analysis}: This task takes a user query and produces an answer using tables in a lake. We use queries in (BIRD, Spider, KramaBench, 
BI-Bench)\footnote{For data discovery and analysis, to avoid contamination from memorization in pre-training, we follow prior work~\cite{annotationerrors, drspider} to use paraphrased queries for all benchmarks.
    \iftoggle{full}
    {
         We also exclude ambiguous queries that can be answered by multiple tables in the lake via LLM-preprocessing.
    }
}, and report \emph{accuracy} \cite{bird,kramabench} by comparing the predicted answer with ground-truth answers from these benchmarks.

\underline{Data science}: Given a table $T$ and a target column $C \in T$, in this task an agent needs to build a predictive-model (regression for numerical targets and classification for categorical targets) to best predict the target column $C$. In a data lake setting, the agent can find additional external tables that may help improve
prediction accuracy by joining or unioning into $T$.
We evaluate using a set of 16 classification and 45 regression tasks\footnote{We construct this benchmark with 16 classification and
45 regression tasks, by selecting base tables with more than 1,000 rows and
declared joins, as well as a meaningful target column.},
where we measure accuracy for classification tasks, and $R^2$ for regression tasks.
We report \emph{gain recovery (GR)}, defined as 
 $\mathrm{GR} = \frac{\overline{\mathit{gain}}_{\mathrm{method}}}{\overline{\mathit{gain}}_{\mathrm{Gold}}}\times100\%$, which measures how much gain a method recovers on average ($\overline{\mathit{gain}}_{\mathrm{method}}$) relative to the average gain of a Gold method  ($\overline{\mathit{gain}}_{\mathrm{Gold}}$) \cite{chaudhuri1997efficient} that integrates all tables known to be related via joins/unions, followed by training using AutoGluon~\cite{erickson2020autogluon} to completion without a time limit.



\paragraph{\textbf{Baselines}}
We compare 12 representative methods from 4 categories below. Because tasks above like data discovery and analysis are interactive in nature,  we place a 4-turn limit where LLMs can interact with the environment at most 4 times, to make latency acceptable, which is fair across all methods (as otherwise we observe some baselines like coding-agents can spend \emph{over an hour} working on a data lake). We additionally compare results with no such turn-limit  in our sensitivity studies.

(1)~\emph{\underline{Cross-table annotation}.}
\textcircled{1} This is our proposed \method method. We report results using the variant that extract relationships in all tables (Section~\ref{sec:offline-step}) by default, and study the variant with budget (Section~\ref{sec:budget-generation}) in Section~\ref{sec:budget-experiments}.

(2)~\emph{\underline{Per-table annotation}.}
\textcircled{1}~AutoDDG~\cite{autoddg} is a representative method in the per-table category, that prompts LLM to generate semantic descriptions per table, without cross-table relationships.

(3)~\emph{\underline{Table retrieval methods}.} Since we are operating on a data lake for table tasks where retrieval of relevant tables is a key step, we compare with existing table-retrieval methods.
We configure each retrieval method in two settings: \emph{\underline{(i) one-shot}} where we feed the original query $q$ into the retrieval system, and directly use the returned top-K tables from each method for downstream tasks (data discovery/analysis/science); and \emph{\underline{(ii) agentic}} where we give an LLM full autonomy to rephrase query $q$ and invoke the retrieval system multiple times as needed, which is similar to the tool-use setting of \method. 

This produces a total of 8 baselines: 
\textcircled{2}~BM25 \cite{robertson1994some} and \textcircled{3}~BM25 (agentic) for lexical retrieval,
\textcircled{4}~Hybrid (RRF) \cite{cormack2009reciprocal} and \textcircled{5}~Hybrid (RRF) (agentic)
for lexical/vector rank fusion,
\textcircled{6}~LlamaIndex \cite{llamaindex} and \textcircled{7}~LlamaIndex (agentic)
for table-vector retrieval,
\textcircled{8}~Pneuma \cite{balaka2025pneuma} and \textcircled{9}~Pneuma (agentic) 
for table summary retrieval.


(4)~\emph{\underline{Coding agents}.} Modern coding agents such as Claude Code~\cite{claudecode}, Codex CLI~\cite{codexcli}, and GitHub Copilot CLI~\cite{githubcopilotcli}, are capable
across a broad range of tasks, making them strong baselines for comparison.
We use GitHub Copilot CLI as a representative frontier coding agent.
We evaluate it using the same 4-turn limit as other methods, labeled as
\textcircled{\scriptsize 10}~\emph{Coding agent}; as well as \textcircled{\scriptsize 11}~\emph{Coding agent (unlimited)}, where we give the coding agent an unfair advantage with unlimited turns, and allow it to explore the lake for as much as it wants. This serves as a really strong reference point for quality comparisons.


\paragraph{\textbf{Backbone models}}
We evaluate all methods with 6 different backbone models, for a total of $12 \times 6 = 72$ model/method combinations. When a backbone is selected, it is used consistently across all steps in each method, including offline processing/indexing and online querying. The models we tested include 3  proprietary models \textcircled{1} GPT-5.4; \textcircled{2}
GPT-5.4-mini and \textcircled{3} GPT-5.4-nano; and 3 open-weight models
\textcircled{4} GPT-OSS-120B, \textcircled{5} DeepSeek-v4 Pro, \textcircled{6} DeepSeek-v4 Flash.
We use default reasoning and temperature settings for all models.

\paragraph{\textbf{Other setups}}
We use OpenAI's \texttt{text-embedding-3-large} for text embeddings.
Local data processing and evaluation run on an Ubuntu 24.04
server with 96~vCPUs (AMD EPYC 7452) and 661~GiB of RAM;
LLM inference uses cloud-hosted APIs.

\subsection{Main Results: Quality Comparisons}
\label{sec:main-results}


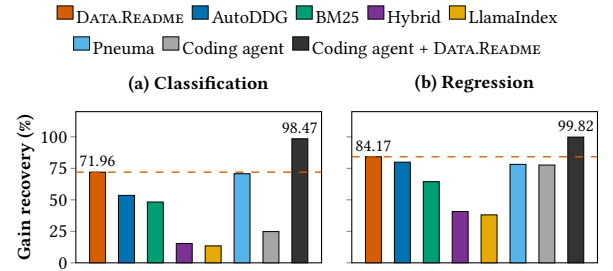
\begin{figure}[H]
  \centering
  \vspace{-4mm}
  
\savebox{\tblbox}{%
\definecolor{dsgreen}{HTML}{009E73}
\definecolor{dsteal}{HTML}{56B4E9}
\definecolor{dsdark}{HTML}{333333}

\definecolor{dsblue}{HTML}{0072B2}
\definecolor{dsred}{HTML}{D55E00}
\definecolor{dspurple}{HTML}{7B3294}
\definecolor{dsgold}{HTML}{E6AB02}
\begin{tabular}{@{}c@{}}
{\fontsize{7}{8}\selectfont \tikz[baseline=-0.5ex]\draw[fill=dsred,draw=black,line width=0.3pt] (0,0) rectangle (6pt,5pt);\,\method{}\hspace{3pt}\tikz[baseline=-0.5ex]\draw[fill=dsblue,draw=black,line width=0.3pt] (0,0) rectangle (6pt,5pt);\,AutoDDG\hspace{3pt}\tikz[baseline=-0.5ex]\draw[fill=dsgreen,draw=black,line width=0.3pt] (0,0) rectangle (6pt,5pt);\,BM25\hspace{3pt}\tikz[baseline=-0.5ex]\draw[fill=dspurple,draw=black,line width=0.3pt] (0,0) rectangle (6pt,5pt);\,Hybrid\hspace{3pt}\tikz[baseline=-0.5ex]\draw[fill=dsgold,draw=black,line width=0.3pt] (0,0) rectangle (6pt,5pt);\,LlamaIndex}\\[-1pt]
{\fontsize{7}{8}\selectfont \tikz[baseline=-0.5ex]\draw[fill=dsteal,draw=black,line width=0.3pt] (0,0) rectangle (6pt,5pt);\,Pneuma\hspace{3pt}\tikz[baseline=-0.5ex]\draw[fill=gray!70,draw=black,line width=0.3pt] (0,0) rectangle (6pt,5pt);\,Coding agent\hspace{3pt}\tikz[baseline=-0.5ex]\draw[fill=dsdark,draw=black,line width=0.3pt] (0,0) rectangle (6pt,5pt);\,Coding agent + \method{}}\\[3pt]
\begin{tikzpicture}[font=\fontsize{7}{8}\selectfont]
\begin{groupplot}[
group style={group size=2 by 1,horizontal sep=4mm,ylabels at=edge left,yticklabels at=edge left},
width=3.25cm,height=2cm,scale only axis,
xmin=0.25,xmax=8.75,xtick=\empty,
ymin=0,ymax=120,enlarge y limits=false,
ytick distance=25,ylabel={Gain recovery (\%)},/pgf/bar width=6pt,
axis background/.style={fill=white},axis line style={black,line width=0.4pt},
axis lines=box,tick align=outside,major tick length=2pt,xtick pos=bottom,ytick pos=left,
tick label style={font=\fontsize{7}{8}\selectfont,inner sep=1pt},
label style={font=\fontsize{7}{8}\selectfont\bfseries,inner sep=1pt},
title style={font=\fontsize{7}{8}\selectfont\bfseries,inner sep=1pt},
scaled y ticks=false,grid=none,clip=false,
point meta=explicit symbolic,nodes near coords={\pgfplotspointmeta},
every node near coord/.append style={font=\fontsize{6.5}{7}\selectfont,text=black,inner sep=1pt,anchor=south,yshift=1pt},
]
\nextgroupplot[title={(a) Classification}]
\addplot+[ybar,bar shift=0pt,solid,fill=dsred,draw=black,line width=0.3pt,mark=none] coordinates {(1,71.95685291697468)[71.96]};
\addplot+[ybar,bar shift=0pt,solid,fill=dsblue,draw=black,line width=0.3pt,mark=none] coordinates {(2,53.53513315996003)[]};
\addplot+[ybar,bar shift=0pt,solid,fill=dsgreen,draw=black,line width=0.3pt,mark=none] coordinates {(3,48.27421402690724)[]};
\addplot+[ybar,bar shift=0pt,solid,fill=dspurple,draw=black,line width=0.3pt,mark=none] coordinates {(4,15.368521218320385)[]};
\addplot+[ybar,bar shift=0pt,solid,fill=dsgold,draw=black,line width=0.3pt,mark=none] coordinates {(5,13.458402488370162)[]};
\addplot+[ybar,bar shift=0pt,solid,fill=dsteal,draw=black,line width=0.3pt,mark=none] coordinates {(6,70.73614568651493)[]};
\addplot+[ybar,bar shift=0pt,solid,fill=gray!70,draw=black,line width=0.3pt,mark=none] coordinates {(7,24.80459658131115)[]};
\addplot+[ybar,bar shift=0pt,solid,fill=dsdark,draw=black,line width=0.3pt,mark=none] coordinates {(8,98.47011986989425)[98.47]};
\addplot[dsred,line width=0.5pt,dashed,mark=none,nodes near coords={}] coordinates {(0.25,71.95685291697468) (8.75,71.95685291697468)};
\nextgroupplot[title={(b) Regression}]
\addplot+[ybar,bar shift=0pt,solid,fill=dsred,draw=black,line width=0.3pt,mark=none] coordinates {(1,84.16973774583808)[84.17]};
\addplot+[ybar,bar shift=0pt,solid,fill=dsblue,draw=black,line width=0.3pt,mark=none] coordinates {(2,79.88777186507878)[]};
\addplot+[ybar,bar shift=0pt,solid,fill=dsgreen,draw=black,line width=0.3pt,mark=none] coordinates {(3,64.42470507566743)[]};
\addplot+[ybar,bar shift=0pt,solid,fill=dspurple,draw=black,line width=0.3pt,mark=none] coordinates {(4,40.700259059417036)[]};
\addplot+[ybar,bar shift=0pt,solid,fill=dsgold,draw=black,line width=0.3pt,mark=none] coordinates {(5,38.06268388643234)[]};
\addplot+[ybar,bar shift=0pt,solid,fill=dsteal,draw=black,line width=0.3pt,mark=none] coordinates {(6,78.16685859339874)[]};
\addplot+[ybar,bar shift=0pt,solid,fill=gray!70,draw=black,line width=0.3pt,mark=none] coordinates {(7,77.6238220056823)[]};
\addplot+[ybar,bar shift=0pt,solid,fill=dsdark,draw=black,line width=0.3pt,mark=none] coordinates {(8,99.8189948688358)[99.82]};
\addplot[dsred,line width=0.5pt,dashed,mark=none,nodes near coords={}] coordinates {(0.25,84.16973774583808) (8.75,84.16973774583808)};
\end{groupplot}
\end{tikzpicture}
\end{tabular}
}
\usebox{\tblbox}

  \vspace{-4mm}
  \caption{Data Science: average gain recovery (\%)}
  \vspace{-4mm}
  \label{fig:datascience}
  \label{fig:datascience-side-by-side}
\end{figure}

\begin{figure}[!t]
  \centering
  \includegraphics[width=\linewidth]{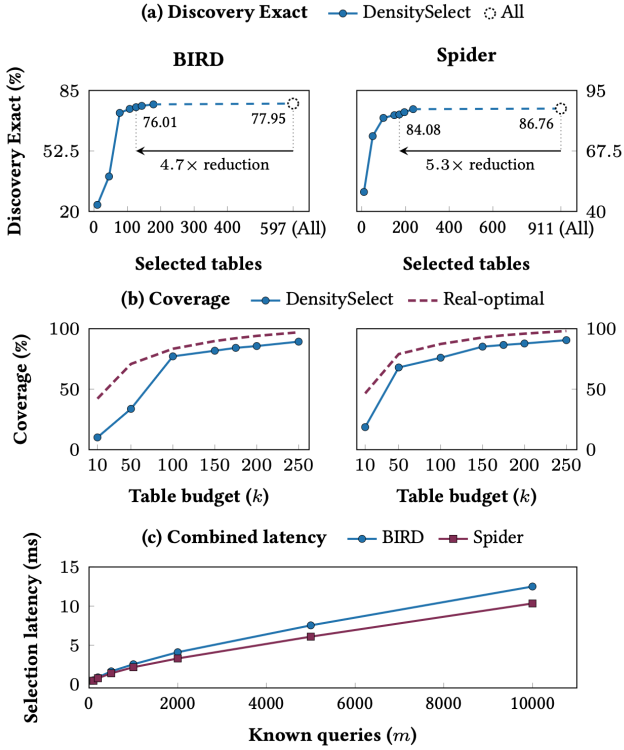}
  \vspace{-4mm}
  \caption{DensitySelect scalability vs. budget.}
  \vspace{-4mm}
  \label{fig:budget-knapsack-dataset-full}
  \label{fig:budget-generation}
  \label{fig:budget-physical-coverage}
\end{figure}

\smallskip\emph{\underline{Data discovery}.}
Table~\ref{tab:discovery-main} shows a comparison of all methods on quality, across all models (6), benchmarks (4), and lakes (2) combinations, for a total of $6 \times 4 \times 2 = 48$ settings.  \method outperforms all comparable baselines with the same turn-limit in 47 of 48 settings (with the remaining one being a tie). It significantly outperforms per-table annotation methods like AutoDDG, existing retrieval systems, and coding-agents.
\iftoggle{full}
{
}
{
    In the interest of space, we report detailed statistical significance of these comparisons in an extended version of this paper~\cite{full}.
}

\emph{\lakeaugmented{}} is clearly more challenging, where all methods have lower scores.  \method continues to show a substantial lead over all methods on this lake.
Even against the ``Coding agent (unlimited)'' reference point (which enjoys an advantage of having unlimited LLM turns), \method
achieves higher accuracy in 13 of the 16 settings.



\smallskip\emph{\underline{Data analysis}.}
Table~\ref{tab:analysis-main} shows a similar comparison for data analysis, where \method achieves the highest
score 
in 45 of 48 settings. 
As with data discovery,
\method generally outperforms the strong Coding agent (unlimited). It is encouraging to see that  \method using
smaller backbones can often outperform baselines using
stronger and more expensive backbones. We also note that on the challenging \emph{\lakeaugmented{}}, the lead of \method
often widens compared to \emph{\lakeoriginal{}{}}.

\begin{table*}[tp]
    \vspace{-18mm}
  \caption{Data discovery exact match (\%). Each cell: \lakeoriginal{} / \lakeaugmented{}.
   For Coding agent, --/-- denotes unavailable results because the
   corresponding model is not offered in Coding agent as a model option.}
  \label{tab:discovery-main}
  \vspace{-2mm}
  
\savebox{\tblbox}{%
\renewcommand{\best}[1]{\textcolor{red}{\textbf{#1}}}
\begin{tabular}{lrrrrrrrrrrr@{\hspace{12pt}}|@{\hspace{12pt}}r}
\toprule
model & {\scriptsize\shortstack{\method{}}} & {\scriptsize\shortstack{AutoDDG}} & {\scriptsize\shortstack{Pneuma}} & {\scriptsize\shortstack{Pneuma\\(agentic)}} & {\scriptsize\shortstack{BM25}} & {\scriptsize\shortstack{BM25\\(agentic)}} & {\scriptsize\shortstack{Hybrid\\(RRF)}} & {\scriptsize\shortstack{Hybrid\\(RRF)\\(agentic)}} & {\scriptsize\shortstack{LlamaIndex}} & {\scriptsize\shortstack{LlamaIndex\\(agentic)}} & {\scriptsize\shortstack{Coding agent\\(Copilot CLI)}} & {\scriptsize\shortstack{Coding agent\\(Copilot CLI, unlimited)}} \\
\midrule
\multicolumn{12}{l@{\hspace{12pt}}|@{\hspace{12pt}}}{\textbf{bird-paraphrase}} & \\
gpt-5.4 & \best{74.7}/\best{51.7} & 64.5/37.6 & 4.1/0.2 & 50.9/26.4 & 2.3/0.8 & 44.7/28.4 & 13.6/4.4 & 60.1/39.1 & 12.5/1.6 & 58.1/34.3 & 14.8/7.4 & 72.6/53.0 \\
gpt-5.4-mini & \best{57.0}/\best{37.4} & 49.6/20.9 & 5.7/0.7 & 39.6/20.5 & 2.3/0.8 & 30.5/15.1 & 13.6/4.4 & 45.0/26.1 & 12.5/1.6 & 44.0/25.5 & 8.7/4.4 & 50.1/31.4 \\
gpt-5.4-nano & \best{57.3}/\best{35.8} & 48.3/18.2 & 4.4/0.3 & 34.3/14.8 & 2.3/0.8 & 27.8/12.2 & 13.6/4.4 & 41.4/23.8 & 12.5/1.6 & 39.7/23.8 & \na/\na & \na/\na \\
gpt-oss-120b\_1 & \best{41.7}/\best{21.7} & 35.6/15.3 & 4.4/0.3 & 24.0/8.0 & 2.3/0.8 & 23.3/11.2 & 13.6/4.4 & 39.2/\best{21.7} & 12.5/1.6 & 36.8/21.0 & \na/\na & \na/\na \\
deepseek-v4-pro & \best{66.8}/\best{46.3} & 51.2/25.1 & 4.4/0.3 & 36.9/13.3 & 2.3/0.8 & 31.4/19.7 & 13.6/4.4 & 46.1/27.1 & 12.5/1.6 & 47.5/25.3 & \na/\na & \na/\na \\
deepseek-v4-flash & \best{60.8}/\best{37.9} & 49.3/21.8 & 3.0/0.0 & 25.8/9.2 & 2.3/0.8 & 23.5/9.9 & 13.6/4.4 & 45.2/24.0 & 12.5/1.6 & 46.0/22.7 & \na/\na & \na/\na \\
\midrule
\multicolumn{12}{l@{\hspace{12pt}}|@{\hspace{12pt}}}{\textbf{spider-paraphrase}} & \\
gpt-5.4 & \best{75.2}/\best{66.1} & 65.3/51.8 & 10.9/7.7 & 61.7/43.1 & 0.0/0.0 & 31.4/25.2 & 18.2/12.4 & 66.1/55.8 & 37.2/18.2 & 69.0/53.3 & 25.9/12.4 & 69.3/65.7 \\
gpt-5.4-mini & \best{64.6}/\best{56.6} & 58.4/41.2 & 11.3/6.9 & 58.4/36.9 & 0.0/0.0 & 23.7/10.9 & 18.2/12.4 & 60.9/47.4 & 37.2/18.2 & 59.9/46.4 & 16.8/12.4 & 51.1/41.2 \\
gpt-5.4-nano & \best{66.1}/\best{55.1} & 54.4/33.6 & 13.1/8.4 & 51.5/33.9 & 0.0/0.0 & 22.6/10.2 & 18.2/12.4 & 55.8/43.4 & 37.2/18.2 & 54.7/42.3 & \na/\na & \na/\na \\
gpt-oss-120b\_1 & \best{58.0}/\best{42.7} & 48.5/28.1 & 13.1/6.6 & 43.4/29.2 & 0.0/0.0 & 21.5/11.3 & 18.2/12.4 & 48.5/38.7 & 37.2/18.2 & 55.5/41.2 & \na/\na & \na/\na \\
deepseek-v4-pro & \best{69.0}/\best{62.0} & 55.8/40.9 & 12.0/6.6 & 51.8/33.2 & 0.0/0.0 & 26.3/18.6 & 18.2/12.4 & 57.3/47.8 & 37.2/18.2 & 61.3/51.1 & \na/\na & \na/\na \\
deepseek-v4-flash & \best{64.6}/\best{55.8} & 56.9/45.6 & 11.7/6.2 & 43.8/26.6 & 0.0/0.0 & 14.2/6.9 & 18.2/12.4 & 55.5/50.0 & 37.2/18.2 & 60.2/49.3 & \na/\na & \na/\na \\
\midrule
\multicolumn{12}{l@{\hspace{12pt}}|@{\hspace{12pt}}}{\textbf{bibench}} & \\
gpt-5.4 & \best{68.5}/\best{57.4} & 53.7/27.8 & 13.0/0.0 & 55.6/35.2 & 1.9/0.0 & 46.3/27.8 & 20.4/0.0 & 50.0/38.9 & 13.0/0.0 & 44.4/38.9 & 29.6/25.9 & 64.8/51.9 \\
gpt-5.4-mini & \best{51.9}/\best{42.6} & 42.6/25.9 & 13.0/1.9 & 42.6/29.6 & 1.9/0.0 & 42.6/18.5 & 20.4/0.0 & 48.1/31.5 & 13.0/0.0 & 44.4/31.5 & 18.5/9.3 & 55.6/37.0 \\
gpt-5.4-nano & \best{61.1}/\best{40.7} & 40.7/20.4 & 13.0/0.0 & 42.6/24.1 & 1.9/0.0 & 37.0/22.2 & 20.4/0.0 & 40.7/29.6 & 13.0/0.0 & 42.6/27.8 & \na/\na & \na/\na \\
gpt-oss-120b\_1 & \best{59.3}/\best{46.3} & 27.8/14.8 & 16.7/0.0 & 48.1/16.7 & 1.9/0.0 & 40.7/22.2 & 20.4/0.0 & 44.4/29.6 & 13.0/0.0 & 37.0/24.1 & \na/\na & \na/\na \\
deepseek-v4-pro & \best{59.3}/\best{50.0} & 40.7/27.8 & 11.1/0.0 & 51.9/22.2 & 1.9/0.0 & 44.4/22.2 & 20.4/0.0 & 55.6/35.2 & 13.0/0.0 & 44.4/27.8 & \na/\na & \na/\na \\
deepseek-v4-flash & \best{59.3}/\best{50.0} & 44.4/18.5 & 11.1/0.0 & 40.7/18.5 & 1.9/0.0 & 48.1/11.1 & 20.4/0.0 & 51.9/33.3 & 13.0/0.0 & 48.1/31.5 & \na/\na & \na/\na \\
\midrule
\multicolumn{12}{l@{\hspace{12pt}}|@{\hspace{12pt}}}{\textbf{kramabench-paraphrase}} & \\
gpt-5.4 & \best{63.5}/\best{67.1} & 62.4/64.7 & 7.1/5.9 & 50.6/41.2 & 0.0/0.0 & 41.2/22.4 & 16.5/12.9 & 61.2/57.6 & 18.8/17.6 & 60.0/60.0 & 35.3/23.5 & 65.9/63.5 \\
gpt-5.4-mini & \best{60.0}/\best{61.2} & 55.3/51.8 & 3.5/5.9 & 44.7/29.4 & 0.0/0.0 & 32.9/17.6 & 16.5/12.9 & 50.6/51.8 & 18.8/17.6 & 52.9/51.8 & 24.7/20.0 & 57.6/50.6 \\
gpt-5.4-nano & \best{64.7}/\best{60.0} & 54.1/49.4 & 8.2/9.4 & 45.9/40.0 & 0.0/0.0 & 38.8/20.0 & 16.5/12.9 & 60.0/51.8 & 18.8/17.6 & 54.1/48.2 & \na/\na & \na/\na \\
gpt-oss-120b\_1 & \best{55.3}/\best{51.8} & 48.2/40.0 & 9.4/4.7 & 28.2/24.7 & 0.0/0.0 & 30.6/20.0 & 16.5/12.9 & 54.1/47.1 & 18.8/17.6 & 49.4/45.9 & \na/\na & \na/\na \\
deepseek-v4-pro & \best{69.4}/\best{61.2} & 58.8/54.1 & 8.2/8.2 & 44.7/34.1 & 0.0/0.0 & 28.2/11.8 & 16.5/12.9 & 56.5/50.6 & 18.8/17.6 & 61.2/54.1 & \na/\na & \na/\na \\
deepseek-v4-flash & \best{61.2}/\best{63.5} & 52.9/47.1 & 10.6/7.1 & 32.9/28.2 & 0.0/0.0 & 16.5/4.7 & 16.5/12.9 & 57.6/49.4 & 18.8/17.6 & 55.3/54.1 & \na/\na & \na/\na \\
\bottomrule
\end{tabular}}
\usebox{\tblbox}

\readmetablenotes{%
\vspace{6pt}
\parbox{\wd\tblbox}{\footnotesize \textbf{Discovery exact-set match (\%): \lakeoriginal{} / \lakeaugmented{}.} Every cell reports \lakeoriginal{} (4,127 tables) on the left and \lakeaugmented{} (12,829 tables) on the right, never recall/exact. Red maxima are determined independently within each lake using unrounded scores, preserving ties and excluding Coding agent (unlimited). Lakes are not averaged. Each side retains its own missing-value markers, evaluation population and distinct-measurement missing-answer accounting. \textbf{discovery --- \lakeoriginal{}}. exact-set match (percent). Evaluation subsets: bird-paraphrase n=609, spider-paraphrase n=274, bibench n=54, kramabench-paraphrase n=85. Only tool-equipped metadata readers are shown; \method{} readers start from top50. Retrieval and agentic variants are separate methods. BM25, Hybrid (RRF) and LlamaIndex retrieval scores are model-independent and repeated for row alignment, not separate experiments. Pneuma retrieval uses model-specific offline summaries; Pneuma (agentic) combines native judged retrieval with the tools-enabled reader, retaining the original relevance judge for initial and paged retrieval; historical judge-only scores are not substituted. Hybrid uses BM25 + vector RRF. The unqualified Hybrid and Pneuma retrieval columns use raw top50; BM25 and LlamaIndex retrieval use top300. The coding agents use GPT-5.4 and GPT-5.4-mini on mixed-format lakes, not metadata-reader interfaces; unmeasured backbones remain unavailable. Unlimited has no work-call quota; 4 calls admits four individual work calls, including failures, without advance notice, plus separate answer-only submission. \method{} instead allows four successful metadata calls. Only completed, audited full-queryload results are shown, never partial or small100 estimates. Both native backbones now cover all four workloads in both settings. Mini measurements are fresh, never reused GPT-5.4 scores. Coding agent GPT-5.4-mini prefix4, spider: 1 no-answer/timeout runs retain zero scores in the complete denominator. Coding agent GPT-5.4-mini unlimited, bird: 1 no-answer/timeout runs retain zero scores in the complete denominator. Coding agent GPT-5.4-mini unlimited, spider: 1 no-answer/timeout runs retain zero scores in the complete denominator. Coding agent unlimited, bird: 2 no-answer/timeout runs retain zero scores in the complete denominator. Coding agent unlimited, kramabench\_para: 2 no-answer/timeout runs retain zero scores in the complete denominator. Coding agent unlimited, spider: 2 no-answer/timeout runs retain zero scores in the complete denominator. Pneuma (agentic) now covers all six measured backbones. The original three-model Pneuma campaign retains 419 failed records as zeros across 24 complete cells (6,132 question-level results); these are not missing rows. Nano Discovery uses the accepted paygo-discovery-v2 campaign output paths and shared non-ambiguous rescoring, never raw full-load sidecar averages. Failed records retain zero scores without rerunning or removing questions. Nano Pneuma (agentic) retains the native relevance judge for initial and paged retrieval; the unqualified Pneuma column remains raw fusion. All 125 failed records in Nano's 72-cell campaign (including the hidden BASE condition) remain zero-scored. GPT-5.4 Hybrid (agentic) on Krama uses rerun-20260919-v2-reportable85 (52/85 exact), retaining one failed answer as zero. Except for the newly populated Nano, DeepSeek and Pneuma (agentic) cells and this Hybrid rerun, metadata-reader scores retain the published 6d12c8e snapshot. The main comparison omits BASE and the no (table names) control; both remain in full diagnostic reports. \method{} (agentic) denotes the complete U/C/J tools condition, and AutoDDG (agentic) retains its tools condition. Coding agent (unlimited) is a reference-only column, separated on the far right because its work-call budget is not matched to the other methods. It is excluded from best-value highlighting. Best among the remaining methods in each row is \best{red}. Grey dashes mean unavailable; blank cells are not applicable. Missing answers retain zero in the fixed denominator; 13 missing result rows across distinct measurements. DeepSeek Flash/Pro use the accepted 336-cell matrix: 144 Discovery cells from paygo-discovery-v6 and 192 Analysis cells from paygo-analysis-v5, including hidden BASE and no (table names) controls. They retain their own model-specific measurements, fixed task populations and original failure zeros. Logged empty outputs and authorized filtered-empty outcomes remain in the denominator; unavailable historical tool counts are not invented. Flash retains its recorded deployment transitions and authorized continuations, not a single-version rerun. DeepSeek uses the recorded 8,192-token cap with provider default reasoning, not a matched explicit reasoning-effort setting. Native DeepSeek Coding agent measurements remain unavailable.  \textbf{discovery --- \lakeaugmented{}}. exact-set match (percent). Evaluation subsets: bird-paraphrase n=609, spider-paraphrase n=274, bibench n=54, kramabench-paraphrase n=85. Only tool-equipped metadata readers are shown; \method{} readers start from top50. Retrieval and agentic variants are separate methods. BM25, Hybrid (RRF) and LlamaIndex retrieval scores are model-independent and repeated for row alignment, not separate experiments. Pneuma retrieval uses model-specific offline summaries; Pneuma (agentic) combines native judged retrieval with the tools-enabled reader, retaining the original relevance judge for initial and paged retrieval; historical judge-only scores are not substituted. Hybrid uses BM25 + vector RRF. The unqualified Hybrid and Pneuma retrieval columns use raw top50; BM25 and LlamaIndex retrieval use top300. The coding agents use GPT-5.4 and GPT-5.4-mini on mixed-format lakes, not metadata-reader interfaces; unmeasured backbones remain unavailable. Unlimited has no work-call quota; 4 calls admits four individual work calls, including failures, without advance notice, plus separate answer-only submission. \method{} instead allows four successful metadata calls. Only completed, audited full-queryload results are shown, never partial or small100 estimates. Both native backbones now cover all four workloads in both settings. Mini measurements are fresh, never reused GPT-5.4 scores. Coding agent GPT-5.4-mini prefix4, kramabench\_para: 1 no-answer/timeout runs retain zero scores in the complete denominator. Coding agent GPT-5.4-mini unlimited, bird: 3 no-answer/timeout runs retain zero scores in the complete denominator. Coding agent prefix4, bird: 1 no-answer/timeout runs retain zero scores in the complete denominator. Coding agent unlimited, bird: 3 no-answer/timeout runs retain zero scores in the complete denominator. Pneuma (agentic) now covers all six measured backbones. The original three-model Pneuma campaign retains 419 failed records as zeros across 24 complete cells (6,132 question-level results); these are not missing rows. Nano Discovery uses the accepted paygo-discovery-v2 campaign output paths and shared non-ambiguous rescoring, never raw full-load sidecar averages. Failed records retain zero scores without rerunning or removing questions. Nano Pneuma (agentic) retains the native relevance judge for initial and paged retrieval; the unqualified Pneuma column remains raw fusion. All 125 failed records in Nano's 72-cell campaign (including the hidden BASE condition) remain zero-scored. Except for the newly populated Nano, DeepSeek and Pneuma (agentic) cells, metadata-reader scores retain the published 6d12c8e snapshot. The main comparison omits BASE and the no (table names) control; both remain in full diagnostic reports. \method{} (agentic) denotes the complete U/C/J tools condition, and AutoDDG (agentic) retains its tools condition. Coding agent (unlimited) is a reference-only column, separated on the far right because its work-call budget is not matched to the other methods. It is excluded from best-value highlighting. Best among the remaining methods in each row is \best{red}. Grey dashes mean unavailable; blank cells are not applicable. Missing answers retain zero in the fixed denominator; 11 missing result rows across distinct measurements. DeepSeek Flash/Pro use the accepted 336-cell matrix: 144 Discovery cells from paygo-discovery-v6 and 192 Analysis cells from paygo-analysis-v5, including hidden BASE and no (table names) controls. They retain their own model-specific measurements, fixed task populations and original failure zeros. Logged empty outputs and authorized filtered-empty outcomes remain in the denominator; unavailable historical tool counts are not invented. Flash retains its recorded deployment transitions and authorized continuations, not a single-version rerun. DeepSeek uses the recorded 8,192-token cap with provider default reasoning, not a matched explicit reasoning-effort setting. Native DeepSeek Coding agent measurements remain unavailable. }
}

  \par\vspace{-2mm}
\end{table*}

\begin{table*}[tp]
    \vspace{-18mm}
  \caption{Data analysis scores (\%). Each cell: \lakeoriginal{} / \lakeaugmented{}.
   For Coding agent, --/-- denotes unavailable results because the
   corresponding model is not offered in Coding agent as a model option.}
  \label{tab:analysis-main}
  \vspace{-2mm}
  
\savebox{\tblbox}{%
\renewcommand{\best}[1]{\textcolor{red}{\textbf{#1}}}
\begin{tabular}{lrrrrrrrrrrr@{\hspace{12pt}}|@{\hspace{12pt}}r}
\toprule
model & {\scriptsize\shortstack{\method{}}} & {\scriptsize\shortstack{AutoDDG}} & {\scriptsize\shortstack{Pneuma}} & {\scriptsize\shortstack{Pneuma\\(agentic)}} & {\scriptsize\shortstack{BM25}} & {\scriptsize\shortstack{BM25\\(agentic)}} & {\scriptsize\shortstack{Hybrid\\(RRF)}} & {\scriptsize\shortstack{Hybrid\\(RRF)\\(agentic)}} & {\scriptsize\shortstack{LlamaIndex}} & {\scriptsize\shortstack{LlamaIndex\\(agentic)}} & {\scriptsize\shortstack{Coding agent\\(Copilot CLI)}} & {\scriptsize\shortstack{Coding agent\\(Copilot CLI, unlimited)}} \\
\midrule
\multicolumn{12}{l@{\hspace{12pt}}|@{\hspace{12pt}}}{\textbf{bird-paraphrase (NL2SQL)}} & \\
gpt-5.4 & \best{43.3}/\best{43.8} & 36.7/31.1 & 19.7/13.2 & 30.2/22.1 & 16.0/10.7 & 26.9/22.1 & 26.1/19.5 & 36.2/26.6 & 23.6/15.7 & 34.0/25.8 & 4.8/5.3 & 28.9/26.9 \\
gpt-5.4-mini & \best{32.1}/\best{32.4} & 28.9/22.5 & 16.7/11.2 & 21.5/14.7 & 12.9/8.6 & 19.5/12.1 & 20.7/15.0 & 26.3/19.7 & 17.0/13.9 & 24.6/15.9 & 4.5/5.3 & 23.5/19.0 \\
gpt-5.4-nano & \best{31.4}/\best{29.6} & 26.0/19.2 & 10.9/8.4 & 14.5/9.4 & 11.7/7.9 & 15.4/10.9 & 16.2/12.2 & 22.8/16.9 & 16.0/12.1 & 21.8/15.5 & \na/\na & \na/\na \\
gpt-oss-120b\_1 & \best{24.0}/\best{19.3} & 9.1/8.3 & 12.4/8.4 & 12.9/9.4 & 11.4/7.3 & 14.0/10.7 & 16.7/12.2 & 19.3/13.6 & 15.7/11.9 & 17.9/11.6 & \na/\na & \na/\na \\
deepseek-v4-pro & \best{37.9}/\best{39.0} & 28.9/23.6 & 8.6/7.1 & 20.0/14.4 & 9.1/6.1 & 19.7/11.7 & 11.9/10.1 & 25.6/21.0 & 10.7/8.8 & 24.8/17.4 & \na/\na & \na/\na \\
deepseek-v4-flash & \best{33.6}/\best{32.6} & 26.6/19.7 & 7.3/6.0 & 15.5/9.4 & 8.6/5.8 & 15.9/9.8 & 12.1/8.6 & 25.6/17.9 & 8.8/7.8 & 22.6/17.4 & \na/\na & \na/\na \\
\midrule
\multicolumn{12}{l@{\hspace{12pt}}|@{\hspace{12pt}}}{\textbf{spider-paraphrase (NL2SQL)}} & \\
gpt-5.4 & \best{64.6}/\best{64.6} & 57.3/52.2 & 38.3/33.6 & 51.1/41.2 & 10.9/6.6 & 28.8/21.2 & 52.9/46.0 & 56.2/51.1 & 52.9/47.4 & 59.9/54.4 & 6.2/4.4 & 43.1/41.6 \\
gpt-5.4-mini & \best{57.3}/\best{57.3} & 47.8/47.4 & 36.9/30.7 & 42.3/33.9 & 10.2/6.2 & 22.3/11.7 & 44.5/40.5 & 49.6/46.0 & 45.6/40.1 & 52.2/49.3 & 5.1/4.0 & 33.9/30.7 \\
gpt-5.4-nano & \best{52.2}/\best{51.1} & 45.3/38.0 & 31.4/26.3 & 35.4/28.8 & 9.5/4.7 & 17.2/9.9 & 36.5/33.6 & 46.0/39.1 & 40.1/37.6 & 46.4/43.8 & \na/\na & \na/\na \\
gpt-oss-120b\_1 & \best{44.5}/\best{40.1} & 29.2/28.1 & 33.9/25.5 & 37.6/29.6 & 8.4/5.8 & 18.2/9.5 & 39.8/33.9 & 36.9/35.4 & 40.9/36.5 & 41.2/36.9 & \na/\na & \na/\na \\
deepseek-v4-pro & \best{52.9}/\best{54.7} & 48.2/45.3 & 25.2/21.2 & 34.7/29.2 & 7.7/4.7 & 19.3/10.9 & 23.0/23.4 & 48.9/44.5 & 29.9/28.5 & 52.2/45.6 & \na/\na & \na/\na \\
deepseek-v4-flash & \best{56.2}/\best{50.4} & 46.7/44.5 & 21.9/16.8 & 35.0/24.8 & 5.1/4.4 & 9.1/5.1 & 19.0/20.4 & 45.6/41.2 & 22.3/23.7 & 46.4/40.9 & \na/\na & \na/\na \\
\midrule
\multicolumn{12}{l@{\hspace{12pt}}|@{\hspace{12pt}}}{\textbf{bibench (NL2SQL)}} & \\
gpt-5.4 & \best{63.0}/\best{63.0} & 50.0/37.0 & 40.7/24.1 & 46.3/27.8 & 18.5/13.0 & 40.7/27.8 & 38.9/24.1 & 51.9/42.6 & 33.3/24.1 & 46.3/35.2 & 7.4/3.7 & 51.9/57.4 \\
gpt-5.4-mini & \best{51.9}/\best{51.9} & 38.9/22.2 & 29.6/20.4 & 29.6/24.1 & 20.4/11.1 & 37.0/22.2 & 31.5/20.4 & 40.7/29.6 & 40.7/14.8 & 42.6/24.1 & 1.9/1.9 & 42.6/40.7 \\
gpt-5.4-nano & \best{44.4}/\best{35.2} & 35.2/29.6 & 22.2/13.0 & 27.8/14.8 & 16.7/9.3 & 35.2/20.4 & 31.5/16.7 & 25.9/20.4 & 31.5/14.8 & 27.8/22.2 & \na/\na & \na/\na \\
gpt-oss-120b\_1 & \best{44.4}/\best{40.7} & 9.3/9.3 & 25.9/13.0 & 22.2/7.4 & 16.7/7.4 & 31.5/16.7 & 29.6/22.2 & 29.6/24.1 & 33.3/16.7 & 25.9/20.4 & \na/\na & \na/\na \\
deepseek-v4-pro & \best{57.4}/\best{55.6} & 38.9/29.6 & 13.0/0.0 & 31.5/13.0 & 11.1/5.6 & 31.5/13.0 & 16.7/9.3 & 35.2/29.6 & 22.2/9.3 & 40.7/27.8 & \na/\na & \na/\na \\
deepseek-v4-flash & \best{50.0}/\best{50.0} & 42.6/18.5 & 7.4/1.9 & 27.8/14.8 & 11.1/9.3 & 22.2/14.8 & 13.0/7.4 & 38.9/27.8 & 11.1/5.6 & 37.0/25.9 & \na/\na & \na/\na \\
\midrule
\multicolumn{12}{l@{\hspace{12pt}}|@{\hspace{12pt}}}{\textbf{kramabench-paraphrase (NL2Python)}} & \\
gpt-5.4 & 58.4/\best{59.6} & 56.5/59.2 & 42.2/39.3 & 53.5/44.8 & 21.8/17.6 & 54.2/42.0 & 51.8/52.7 & 57.4/57.9 & 45.7/47.7 & \best{58.7}/57.9 & 21.0/13.3 & 64.5/58.3 \\
gpt-5.4-mini & \best{52.2}/\best{51.9} & 36.8/41.6 & 17.7/18.6 & 38.8/33.4 & 8.9/4.5 & 35.0/16.1 & 25.8/31.5 & 39.8/35.9 & 27.7/26.4 & 43.4/44.9 & 14.9/12.9 & 61.1/50.7 \\
gpt-5.4-nano & \best{38.2}/\best{38.3} & 28.2/25.9 & 21.4/13.6 & 32.7/25.5 & 10.5/6.4 & 26.3/17.6 & 28.1/21.2 & 32.2/33.0 & 21.4/23.0 & 34.6/28.9 & \na/\na & \na/\na \\
gpt-oss-120b\_1 & \best{35.6}/\best{30.8} & 16.0/7.1 & 18.3/13.8 & 24.1/15.0 & 9.0/10.7 & 15.2/13.7 & 22.3/24.8 & 25.2/27.3 & 21.4/19.1 & 29.0/26.9 & \na/\na & \na/\na \\
deepseek-v4-pro & 39.7/\best{41.5} & 37.4/39.0 & 12.5/11.2 & 32.5/35.3 & 7.1/7.4 & 17.7/11.5 & 18.4/16.6 & \best{41.3}/33.6 & 23.7/22.2 & 35.5/35.7 & \na/\na & \na/\na \\
deepseek-v4-flash & 39.3/\best{41.3} & 31.5/21.5 & 15.4/19.1 & 26.3/22.8 & 7.4/3.1 & 22.6/4.9 & 13.5/14.7 & \best{40.2}/33.4 & 14.0/14.0 & 28.8/31.6 & \na/\na & \na/\na \\
\bottomrule
\end{tabular}}
\usebox{\tblbox}

\readmetablenotes{%
\vspace{6pt}
\parbox{\wd\tblbox}{\footnotesize \textbf{Analysis scores (\%): \lakeoriginal{} / \lakeaugmented{}.} Every cell reports \lakeoriginal{} (4,127 tables) on the left and \lakeaugmented{} (12,829 tables) on the right. BIRD/Spider use SQL EX, BI-Bench uses answer-table matching, and Krama uses its official NL2Python answer-type mean; these metrics are not pooled. Red maxima are determined independently within each lake using unrounded scores, preserving ties and excluding Coding agent (unlimited). Lakes are not averaged. Each side retains its own missing-value markers, evaluation population and distinct-measurement missing-answer accounting. \textbf{analysis --- \lakeoriginal{}}. NL2SQL execution accuracy / NL2Python Krama final-answer score (percent). Evaluation subsets: bird-paraphrase n=605, spider-paraphrase n=274, bibench n=54, kramabench-paraphrase n=98. Our metadata-reader columns use tools; \method{} readers start from top50. Unqualified Pneuma, BM25, Hybrid (RRF) and LlamaIndex columns use their full-report single-call condition: recommended top50 plus one external SQL/Python generation, without reader tools. Their (agentic) columns retain tools-enabled readers. Both are model-specific final-answer measurements; discovery retrieval scores are not substituted. The saved Pneuma single-call results use raw fusion without the relevance judge. BI-Bench uses answer-table comparison; BIRD and Spider use gold-SQL EX. Krama uses eligible98 (100 questions minus two geometry exclusions), not discovery's ambiguity subset. Its official answer-type mean is not pooled with SQL accuracy. Tools-enabled \method{} and metadata-baseline Python readers use metadata/sample tools without execution feedback; only the final program is executed, at most once, over the entire CSV lake. They allow four successful calls and 169,656-character complete-JSON tool replies; mini/OSS retain an 8,192-token reasoning-inclusive cap. For the original three backbones, our \method{} (agentic) column uses the original role-plan reader on full eligible98, retaining immutable first25 plus 73 additional eligible questions and two explicit exclusions. Recovery fragments keep native identities; provider refusals remain zero without resampling; interrupted attempts and cost provenance remain recorded. Other Python cells retain metadata-only-v1. These metadata readers require host-cpu-lease-v1. Nano Analysis uses accepted paygo-analysis-v1 outputs and frozen task-specific populations: SQL EX on 605/274/54 questions and Krama answer-type mean on eligible98. Nano Krama UCJ uses krama-role-plan-v1; its other displayed Python readers use legacy-v4. The original-backbone first25/recovery provenance does not apply to these fresh Nano runs. Nano retains the 8,192-token reasoning-inclusive completion cap. Terminal empty/malformed answers retain zero without new model calls; they are not missing result rows. Native Nano Coding agent is unmeasured. Nano's 96-cell campaign (including hidden BASE) retains 55 resolved terminal method failures as zeros. The coding agents use GPT-5.4 and GPT-5.4-mini on mixed-format lakes, not metadata-reader interfaces; unmeasured backbones remain unavailable. Unlimited has no work-call quota; 4 calls admits four individual work calls, including failures, without advance notice, plus separate answer-only submission. \method{} instead allows four successful metadata calls. Only completed, audited full-queryload results are shown, never partial or small100 estimates. Both native backbones now cover all four workloads in both settings. Mini measurements are fresh, never reused GPT-5.4 scores. Coding agent GPT-5.4-mini unlimited, bird: 1 no-answer/timeout runs retain zero scores in the complete denominator. Coding agent unlimited, bird: 1 no-answer/timeout runs retain zero scores in the complete denominator. Except for the new Nano and DeepSeek rows, non-Coding agent scores retain the published 6d12c8e snapshot. The main comparison omits BASE and the no (table names) control; both remain in full diagnostic reports. \method{} (agentic) denotes the complete U/C/J tools condition, and AutoDDG (agentic) retains its tools condition. Coding agent (unlimited) is a reference-only column, separated on the far right because its work-call budget is not matched to the other methods. It is excluded from best-value highlighting. Best among the remaining methods in each row is \best{red}. Grey dashes mean unavailable; blank cells are not applicable. Missing answers retain zero in the fixed denominator; 11 missing result rows across distinct measurements. DeepSeek Flash/Pro use the accepted 336-cell matrix: 144 Discovery cells from paygo-discovery-v6 and 192 Analysis cells from paygo-analysis-v5, including hidden BASE and no (table names) controls. They retain their own model-specific measurements, fixed task populations and original failure zeros. Logged empty outputs and authorized filtered-empty outcomes remain in the denominator; unavailable historical tool counts are not invented. Flash retains its recorded deployment transitions and authorized continuations, not a single-version rerun. DeepSeek uses the recorded 8,192-token cap with provider default reasoning, not a matched explicit reasoning-effort setting. Native DeepSeek Coding agent measurements remain unavailable. DeepSeek Analysis uses primary SQL EX on 605/274/54 questions and the Krama answer-type mean on eligible98. Its UCJ Python reader uses krama-role-plan-v1; hidden BASE uses krama-reliable-v1, and other Python readers use legacy-v4. These runs do not inherit the original backbones' first25/recovery composition. Flash LlamaIndex (agentic) on \lakeoriginal{} Krama instead uses the full eligible98 llamaindex-rerun-20260930-v1 measurement, selected after observing the rerun: 28.8 replaces 47.1, with nine proven empty final answers retained as zeros. It uses the current unversioned Flash pool rather than the historical 0731 pool; the GPT-5.4 rerun is not adopted.  \textbf{analysis --- \lakeaugmented{}}. NL2SQL execution accuracy / NL2Python Krama final-answer score (percent). Evaluation subsets: bird-paraphrase n=605, spider-paraphrase n=274, bibench n=54, kramabench-paraphrase n=98. Our metadata-reader columns use tools; \method{} readers start from top50. Unqualified Pneuma, BM25, Hybrid (RRF) and LlamaIndex columns use their full-report single-call condition: recommended top50 plus one external SQL/Python generation, without reader tools. Their (agentic) columns retain tools-enabled readers. Both are model-specific final-answer measurements; discovery retrieval scores are not substituted. The saved Pneuma single-call results use raw fusion without the relevance judge. BI-Bench uses answer-table comparison; BIRD and Spider use gold-SQL EX. Krama uses eligible98 (100 questions minus two geometry exclusions), not discovery's ambiguity subset. Its official answer-type mean is not pooled with SQL accuracy. Tools-enabled \method{} and metadata-baseline Python readers use metadata/sample tools without execution feedback; only the final program is executed, at most once, over the entire CSV lake. They allow four successful calls and 169,656-character complete-JSON tool replies; mini/OSS retain an 8,192-token reasoning-inclusive cap. For the original three backbones, our \method{} (agentic) column uses the original role-plan reader on full eligible98, retaining immutable first25 plus 73 additional eligible questions and two explicit exclusions. Recovery fragments keep native identities; provider refusals remain zero without resampling; interrupted attempts and cost provenance remain recorded. Other Python cells retain metadata-only-v1. These metadata readers require host-cpu-lease-v1. Nano Analysis uses accepted paygo-analysis-v1 outputs and frozen task-specific populations: SQL EX on 605/274/54 questions and Krama answer-type mean on eligible98. Nano Krama UCJ uses krama-role-plan-v1; its other displayed Python readers use legacy-v4. The original-backbone first25/recovery provenance does not apply to these fresh Nano runs. Nano retains the 8,192-token reasoning-inclusive completion cap. Terminal empty/malformed answers retain zero without new model calls; they are not missing result rows. Native Nano Coding agent is unmeasured. Nano's 96-cell campaign (including hidden BASE) retains 55 resolved terminal method failures as zeros. The coding agents use GPT-5.4 and GPT-5.4-mini on mixed-format lakes, not metadata-reader interfaces; unmeasured backbones remain unavailable. Unlimited has no work-call quota; 4 calls admits four individual work calls, including failures, without advance notice, plus separate answer-only submission. \method{} instead allows four successful metadata calls. Only completed, audited full-queryload results are shown, never partial or small100 estimates. Both native backbones now cover all four workloads in both settings. Mini measurements are fresh, never reused GPT-5.4 scores. Coding agent GPT-5.4-mini prefix4, spider: 1 no-answer/timeout runs retain zero scores in the complete denominator. Coding agent GPT-5.4-mini unlimited, bird: 4 no-answer/timeout runs retain zero scores in the complete denominator. Coding agent GPT-5.4-mini unlimited, spider: 1 no-answer/timeout runs retain zero scores in the complete denominator. Coding agent GPT-5.4-mini prefix4, bibench: 1 no-answer/timeout runs retain zero scores in the complete denominator. Coding agent unlimited, bird: 8 no-answer/timeout runs retain zero scores in the complete denominator. Coding agent unlimited, spider: 1 no-answer/timeout runs retain zero scores in the complete denominator. Except for the new Nano and DeepSeek rows, non-Coding agent scores retain the published 6d12c8e snapshot. The main comparison omits BASE and the no (table names) control; both remain in full diagnostic reports. \method{} (agentic) denotes the complete U/C/J tools condition, and AutoDDG (agentic) retains its tools condition. Coding agent (unlimited) is a reference-only column, separated on the far right because its work-call budget is not matched to the other methods. It is excluded from best-value highlighting. Best among the remaining methods in each row is \best{red}. Grey dashes mean unavailable; blank cells are not applicable. Missing answers retain zero in the fixed denominator; 27 missing result rows across distinct measurements. DeepSeek Flash/Pro use the accepted 336-cell matrix: 144 Discovery cells from paygo-discovery-v6 and 192 Analysis cells from paygo-analysis-v5, including hidden BASE and no (table names) controls. They retain their own model-specific measurements, fixed task populations and original failure zeros. Logged empty outputs and authorized filtered-empty outcomes remain in the denominator; unavailable historical tool counts are not invented. Flash retains its recorded deployment transitions and authorized continuations, not a single-version rerun. DeepSeek uses the recorded 8,192-token cap with provider default reasoning, not a matched explicit reasoning-effort setting. Native DeepSeek Coding agent measurements remain unavailable. DeepSeek Analysis uses primary SQL EX on 605/274/54 questions and the Krama answer-type mean on eligible98. Its UCJ Python reader uses krama-role-plan-v1; hidden BASE uses krama-reliable-v1, and other Python readers use legacy-v4. These runs do not inherit the original backbones' first25/recovery composition. }
}

  \par\vspace{-4mm}
\end{table*}

\smallskip\emph{\underline{Data science}.}
\label{sec:beyond-query-answering}
 Figure~\ref{fig:datascience} shows the quality comparison on data science tasks, where  
\method helps agents perform better on predictive
modeling by leveraging related tables -- it recovers 71.96\% of gold method's
gains on classification, and 84.17\% on regression.

We additionally tested a method labeled as ``Coding agent + \method'', shown as the rightmost method in the figure, where we provide our produced \method as an additional resource for the coding-agent. Compared to ``Coding agent'' without readme shown to its left, this readme-augmented coding-agent achieves substantial quality improvements -- its gain recovery increases from 24.80\% to
98.47\% on classification and from 77.62\% to 99.82\% on regression, again showing the benefit of our using readme.

\subsection{Main Results: with Budget-Constraints}
\label{sec:budget-experiments}

Figure~\ref{fig:budget-generation} evaluates our proposed method, DensitySelect
(Algorithm~\ref{alg:budget-fast}), for the budget constraint problem variant\footnote{We construct a query log by sampling 1000 queries from Bird/Spider with replacement and simulating a power-law distribution with $\alpha$ =1}, across different budget levels, with the table budget $k=\frac{B}{\bar{\kappa}}\in\{10,50,100,150,175,200,250\}$.  We compare results from these constrained configurations, against the full unconstrained version (\emph{All}). 

\smallskip\emph{\underline{Discovery quality in end-to-end tasks}.} Figure~\ref{fig:budget-generation}(a) shows that we selectively process tables while retaining near-full discovery quality with
substantially lower construction cost.
At $k=175$, documenting 126 BIRD tables yields a 76.01\% weighted Exact score, only losing 1.94\% from All; whereas construction cost falls from \$104.20 to \$22.17.
On Spider, 172 tables yield 84.08\% versus 86.76\%, a loss of
2.68 points, with cost falling from \$141.23 to \$26.64, showing that DensitySelect can indeed substantially reduce cost while maintaining quality.

\smallskip\emph{\underline{Solution quality using the objective function}.}
Figure~\ref{fig:budget-generation}(b) shows that DensitySelect can achieve near-optimal selection utility as defined in Equation~\eqref{eq:objective}. The dotted line is the true optimal utility (from brute-force enumeration), and the solid line is the utility of the solution produced by our algorithm, where our solution is close to theoretical optimum, showing the effectiveness of our approximation algorithm when measured by the raw objective function.


\smallskip\emph{\underline{Solution scalability}.}
Figure~\ref{fig:budget-generation}(c) shows that our algorithm runs efficiently and can scale to large query workloads as the query size grows -- With $k=200$, median
selection latency at $m=10000$ is 12.50\,ms on BIRD and
10.34\,ms on Spider. The trend with $m$ is consistent with the
$O(m\log m+m\ell)$  complexity discussed in
Section~\ref{sec:budget-algorithm}.

\begin{figure*}[!t]
  \centering
  \input{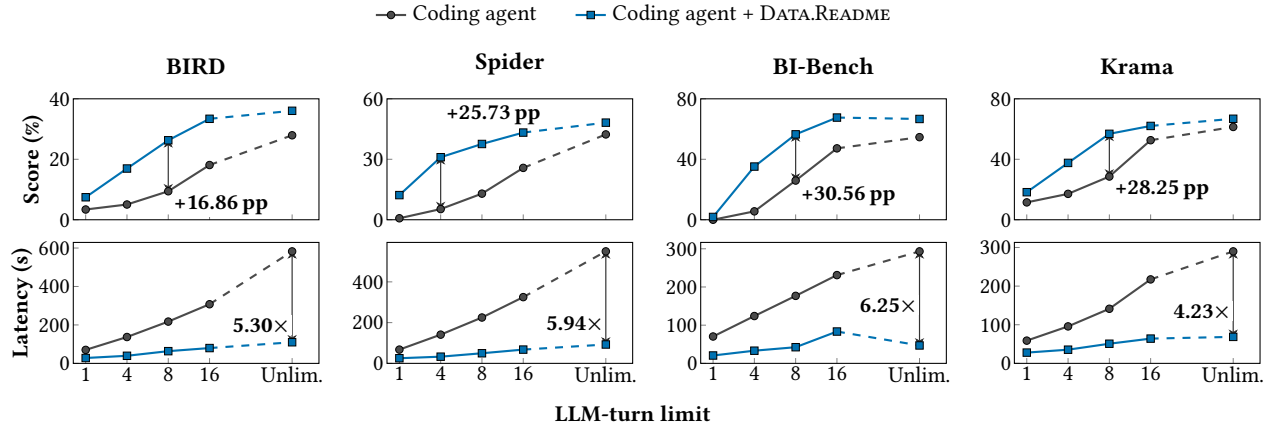}
  \vspace{-4mm}
  \caption{Coding-agent quality and latency with and without \method.}
  \vspace{-4mm}
  \label{fig:coding-agent-sensitivity}
\end{figure*}

\subsection{Ablation Study}
\label{sec:coding-agents}

\textbf{Coding Agent Ablation: with/without Readme}.
Given the strength and generality of today's coding agents, we perform an additional test where we provide our produced \method to coding agent as an additional resource, and test its quality/latency comparisons with vs. without Readme, across different LLM-turn budget settings (1/4/8/16/unlimited).



\smallskip\emph{\underline{Answer quality}.}
Figure~\ref{fig:coding-agent-sensitivity} shows that merely providing readme to a coding agent improves its answer quality at every tested budget setting across all four
workloads. With 8 turns, \emph{Coding agent} with readme improves its counterpart without readme significantly by 16 to 30 percentage points, even when both variants have unlimited turns, the one with readme still retains sizable benefits. 

\smallskip\emph{\underline{Query latency}.}
Importantly, the quality gains of readme are accompanied by higher-efficiency / reduced-latency in
every setting, as readme allows the agent to perform tasks much more efficiently. With unlimited calls, the mean latency of using readme improve by up to 6x times, which represent significantly lower user wait time and better user experience.

\textbf{Ablation of different relationships.}
In Figure~\ref{fig:ablation-summary}, we evaluate how each relationship type extracted contributes to downstream task quality. Starting from Auto-DDG-like per-table table descriptions
and profiles (referred to as Base), we add  confusable tables (+C),
joinable tables (+J),  union-able tables (+U), as well as everything combined (+UCJ) and compare their quality results on benchmark queries involving multiple tables. 
%

The full method (+UCJ) achieves the best  on all four workloads across tasks. The strongest individual
component varies: +J leads on the SQL workloads, while +U leads
or ties on KramaBench, which are dependent on the characteristics of the underlying workload.

\begin{figure}[!htbp]
  \makebox[\columnwidth][l]{%
    \hspace*{-2.5mm}\begin{minipage}{\dimexpr\columnwidth+1mm\relax}
      \savebox{\tblbox}{%

\definecolor{ablblue}{HTML}{0072B2}%
\definecolor{ablgreen}{HTML}{009E73}%
\definecolor{ablpurple}{HTML}{7B3294}%
\definecolor{ablred}{HTML}{D55E00}%
\begin{tabular}{@{}c@{}}
{\fontsize{7}{8}\selectfont \tikz[baseline=-0.5ex]\draw[fill=gray!55,draw=black,line width=0.3pt] (0,0) rectangle (6pt,5pt);\,Base\hspace{6pt}\tikz[baseline=-0.5ex]\draw[fill=ablblue,draw=black,line width=0.3pt] (0,0) rectangle (6pt,5pt);\,+U\hspace{6pt}\tikz[baseline=-0.5ex]\draw[fill=ablgreen,draw=black,line width=0.3pt] (0,0) rectangle (6pt,5pt);\,+C\hspace{6pt}\tikz[baseline=-0.5ex]\draw[fill=ablpurple,draw=black,line width=0.3pt] (0,0) rectangle (6pt,5pt);\,+J\hspace{6pt}\tikz[baseline=-0.5ex]\draw[fill=ablred,draw=black,line width=0.3pt] (0,0) rectangle (6pt,5pt);\,+UCJ}\\[3pt]
\begin{tikzpicture}[font=\fontsize{7}{8}\selectfont]
\begin{groupplot}[
group style={group size=2 by 1,horizontal sep=1mm,
ylabels at=edge left,yticklabels at=edge left},
width=3.8cm,height=2cm,scale only axis,
xmin=0,xmax=29.8,
xtick={3.2,10.7,18.2,25.7},
xticklabels={BIRD,Spider,BI-Bench,Krama},
ymin=30,ymax=70,enlarge y limits=false,ytick distance=10,
ylabel={Score (\%)},/pgf/bar width=3.5pt,
axis background/.style={fill=white},axis line style={black,line width=0.4pt},
axis lines=box,tick align=outside,major tick length=2pt,xtick pos=bottom,ytick pos=left,
tick label style={font=\fontsize{7}{8}\selectfont,inner sep=1pt},
label style={font=\fontsize{7}{8}\selectfont\bfseries,inner sep=1pt},
title style={font=\fontsize{7}{8}\selectfont\bfseries,inner sep=1pt},
scaled y ticks=false,grid=none,clip=true,clip mode=individual,
point meta=explicit symbolic,nodes near coords={\pgfplotspointmeta},
every node near coord/.append style={font=\fontsize{6.5}{7}\selectfont,
text=black,inner sep=1pt,anchor=south,xshift=-3.5pt,yshift=1pt},
]
\nextgroupplot[title={(a) Discovery}]
\addplot+[ybar,bar shift=0pt,solid,fill=gray!55,draw=black,line width=0.3pt,mark=none] coordinates {(0.8,36.36363636363637)[] (8.3,44.73684210526316)[] (15.8,48.97959183673469)[] (23.3,50.98039215686274)[]};
\addplot+[ybar,bar shift=0pt,solid,fill=ablblue,draw=black,line width=0.3pt,mark=none] coordinates {(2.0,42.97520661157025)[] (9.5,45.6140350877193)[] (17.0,57.142857142857146)[] (24.5,54.90196078431372)[]};
\addplot+[ybar,bar shift=0pt,solid,fill=ablgreen,draw=black,line width=0.3pt,mark=none] coordinates {(3.2,37.8099173553719)[] (10.7,42.10526315789474)[] (18.2,53.06122448979592)[] (25.7,54.90196078431372)[]};
\addplot+[ybar,bar shift=0pt,solid,fill=ablpurple,draw=black,line width=0.3pt,mark=none] coordinates {(4.4,48.14049586776859)[] (11.9,57.01754385964912)[] (19.4,59.183673469387756)[] (26.9,52.94117647058823)[]};
\addplot+[ybar,bar shift=0pt,solid,fill=ablred,draw=black,line width=0.3pt,mark=none] coordinates {(5.6,48.34710743801653)[48.35] (13.1,58.771929824561404)[58.77] (20.6,63.265306122448976)[63.27] (28.1,56.86274509803921)[56.86]};
\nextgroupplot[title={(b) Analysis}]
\addplot+[ybar,bar shift=0pt,solid,fill=gray!55,draw=black,line width=0.3pt,mark=none] coordinates {(0.8,34.0956340956341)[] (8.3,42.98245614035088)[] (15.8,51.02040816326531)[] (23.3,51.196859348463136)[]};
\addplot+[ybar,bar shift=0pt,solid,fill=ablblue,draw=black,line width=0.3pt,mark=none] coordinates {(2.0,35.13513513513514)[] (9.5,43.85964912280702)[] (17.0,55.10204081632653)[] (24.5,57.76825633383011)[]};
\addplot+[ybar,bar shift=0pt,solid,fill=ablgreen,draw=black,line width=0.3pt,mark=none] coordinates {(3.2,36.17463617463618)[] (10.7,48.24561403508772)[] (18.2,55.10204081632653)[] (25.7,46.29136456740807)[]};
\addplot+[ybar,bar shift=0pt,solid,fill=ablpurple,draw=black,line width=0.3pt,mark=none] coordinates {(4.4,38.46153846153846)[] (11.9,50.87719298245614)[] (19.4,61.224489795918366)[] (26.9,53.92909978562161)[]};
\addplot+[ybar,bar shift=0pt,solid,fill=ablred,draw=black,line width=0.3pt,mark=none] coordinates {(5.6,42.41164241164241)[42.41] (13.1,53.50877192982456)[53.51] (20.6,61.224489795918366)[61.22] (28.1,59.40611866576871)[59.41]};
\end{groupplot}
\end{tikzpicture}
\end{tabular}%
}
\usebox{\tblbox}

\readmetablenotes{%
\vspace{6pt}
\parbox{\wd\tblbox}{\footnotesize \textbf{GPT-5.4 additive ablation by lake and workload (\%)}. Base / +U / +C / +J / +UCJ use the same task-eligible GT$>1$ questions, top50 retrieval and four-successful-call metadata-tool budget. Each single-component arm adds only that component to Base. Neither workloads nor lakes are averaged. Discovery reports exact@K: truncate to K physical tables before original-file grouping for Krama. Ranking order can therefore change exactness. Analysis reports SQL EX for BIRD/Spider/BI-Bench and the fractional Krama final-answer score; never pool these metrics. Denominators in BIRD/Spider/BI-Bench/Krama order: Discovery: 484/114/49/51. Analysis: 481/114/49/61. Python Base retains krama-reliable-v1; +U/+C/+J/+UCJ use krama-role-plan-v1, so Base comparisons do not isolate metadata from reader-profile differences. The full source report makes post-outcome selections of three complete two-lake conditions and one complete augmented-lake condition: BI-Bench Discovery +U and Krama Analysis +UCJ are unmodified reruns; Krama Discovery +U uses role-order guidance, including ordering-only exact gains. \lakeaugmented{} BI-Bench Discovery +UCJ uses joint-evidence guidance: 31/49 versus 30/49 for its unchanged UCJ control (two gains, one loss). These seven replaced cells are not repeat averages, per-query best-of mixtures, held-out gains or evidence of stable causal improvement. Original records and unsuccessful trials are retained. Sources: additive-gt-multitable-20260924 active report, validated through its read-only active-report verifier and byte-identical historical archives; no source checks are waived or experiments rerun. All failure zeros are retained, including 13 missing-answer zeros. Colors identify Base/+U/+C/+J/+UCJ, not winners. Horizontal labels show only +UCJ. Both panels share labelled, data-dependent percentage limits; the lower bound can be nonzero. All native percentages, task-specific denominators, independent maxima and ties are retained in vector comments. No workloads or task metrics are averaged.}
}
    \end{minipage}%
  }
  \vspace{-4mm}
  \caption{Ablation study of \method.}
  \label{fig:ablation-summary}
  \label{tab:ablation-summary}
\end{figure}

\subsection{Sensitivity and Efficiency}
\label{sec:sensitivity}


\underline{Sensitivity to embedding models}.
We replace the default embedding model,
OpenAI's \texttt{text-embedding-\allowbreak 3-large}, with
three lightweight embedding models that can run locally on CPUs: 
MiniLM-L6~\cite{minilm}, BGE-small~\cite{bge_embedding}, and
E5-small~\cite{wang2022text}. Table~\ref{tab:retriever-sensitivity-augmented} shows a comparison on discovery and analysis quality using \lakeaugmented{}.
With lightweight embedding models, \method performance drops by 2-4 percentage points, but is still significantly better than the best baseline, Hybrid RRF, on both tasks, illustrating that our quality gains are not dependent on big embedding models.

\begin{table}[tbp]
\vspace{-0mm}
  \caption{Embedding-model sensitivity (\%).}
  \label{tab:retriever-sensitivity-augmented}
  \vspace{-3mm}
  
\savebox{\tblbox}{%
\begingroup\fontsize{8.5}{10}\selectfont
\setlength{\tabcolsep}{1.5pt}
\begin{minipage}{8.3cm}
\begin{tabular*}{\linewidth}{@{\extracolsep{\fill}}lrr@{}}\toprule
Method & Discovery (\%) & Analysis (\%) \\ \midrule
\method + text-embedding-3-large & 60.56 & 57.74 \\[2pt]
\method + all-MiniLM-L6-v2 & 56.65 & 53.37 \\[2pt]
\method + bge-small-en-v1.5 & 58.58 & 54.56 \\[2pt]
\method + e5-small-v2 & 58.55 & 54.14 \\[2pt]
\midrule
Best baseline - Hybrid RRF agentic & 47.86 & 44.55 \\[2pt]
\bottomrule
\end{tabular*}
\end{minipage}
\endgroup
}
\usebox{\tblbox}

\readmetablenotes{%
\vspace{6pt}
\parbox{\wd\tblbox}{\footnotesize \textbf{Embedding-model sensitivity.} Final Discovery EX is exact-set match, not recall or initial retrieval coverage. Analysis uses SQL EX for BIRD/Spider, accepted-answer-table accuracy for BI-Bench, and native mean final-answer score for Krama. All scores are percentages, for \lakeaugmented{} only. Each task column is the equal-weight macro mean of BIRD, Spider, BI-Bench and Krama scores, not a pooled question-level accuracy. Lakes are not averaged. GPT-5.4 / UCJ retains top50 retrieval and four successful metadata calls. Discovery denominators are 609/274/54/85; Analysis uses 605/274/54/98, in benchmark order, per lake. The complete two-lake Analysis campaign includes undisplayed conditions. The 32 local-encoder Analysis cells contain 8,296 native records and 8,248 scored measurements; the 14 previously completed result files are unchanged. Historical Azure reader outputs rescored on the identical current question sets, not fresh paired runs; historical prompt equivalence is not certified. Best baseline is one comparable GPT-5.4 main-table method selected by the equal-weight mean of its Discovery and Analysis macro scores. Both displayed scores belong to that same method, not per-workload or per-task maxima. Only methods measured on all four workloads in both tasks are eligible; \method{} and unmatched unlimited are excluded, while ordinary Coding agent is eligible. Exact ties use main-table method order and are retained in metadata. The no-\method{} control is omitted from display but retained in source metadata. Reference results are reused, never rerun. Krama's local \method{} readers use role-plan, whereas the historical no-\method{} and external metadata readers retain metadata-only-v1; this is not a strict retriever-only comparison against those references. Local encoders retain their recommended prefixes and native token limits. Differences cannot be attributed to model weights alone; no significance is claimed.}
}

  \par\vspace{-2mm}
\end{table}

\begin{figure}[!t]
\vspace{-18mm}
  \input{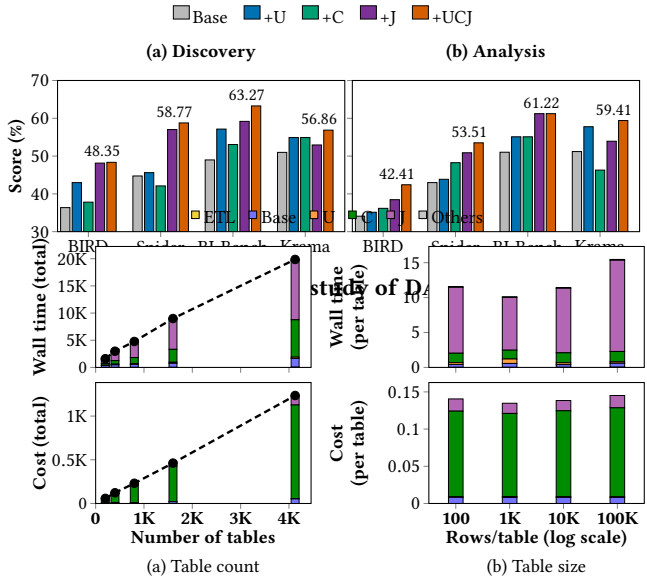}
  \vspace{-2mm}
  \caption{Construction time and estimated cost: (a) total (s, USD);
    (b) per table (s/table, USD/table).}
    \vspace{-6mm}
  \label{fig:construction-scaling}
\end{figure}

\underline{Construction Efficiency and Sensitivity}.
We measure offline construction cost/latency as the number and size
of tables increase. To vary table count, we perform constructions on sampled subsets
with 200, 400, 800, and 1,600 tables from \lakeoriginal{}.
Figure~\ref{fig:construction-scaling}(a) reports elapsed construction time
and LLM cost as we increase the number of tables, which grows linearly with the number of tables.

We also vary the table size (the number of rows in a table), by 
selecting 100 tables from the lake each with at least 100K rows, 
and then down-sample each to 100, 1K, 10K, and 100K rows, respectively, and perform construction experiments. 
Figure~\ref{fig:construction-scaling}(b) shows that LLM cost and latency are not sensitive to row count:
increasing row count from 100 to 100K raises cost and latency only slightly, which shows the scalability of our approach to large tables and large lakes.



\iftoggle{full}
{
    \subsection{Incorporating code-to-table relationships}
    \label{sec:online-setting}
    
    \begin{figure}[t]
      \input{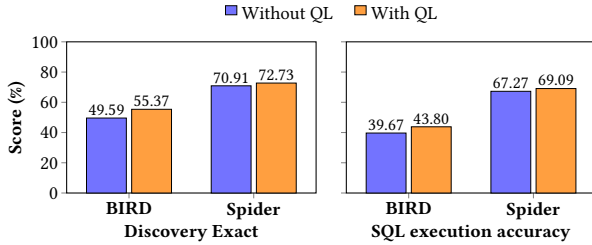}
      \vspace{-4mm}
      \caption{Query history improves data discovery and analysis.}
      \vspace{-4mm}
      \label{fig:online-study}
      \label{fig:online-setting}
      \label{fig:online-analysis}
      \label{fig:online-sensitivity}
    \end{figure}
    
    Historical SQL queries complement table metadata by providing examples of how tables are used ~\cite{gao2023text}, and can be attached to each table using the last column shown in Figure~\ref{fig:system-overview}. We evaluate whether \method{} 
    benefits from this additional context on BIRD and Spider in \lakeaugmented{}. We use 80\% of the queries to construct a historical query log $QL$ and reserve the remaining 20\% for evaluation. To integrate the knowledge of $QL$, We store each SQL in $QL$ in the \texttt{snippet} field of the \method{} entry for every table it references.    These examples are then available alongside retrieved \method{}. 
    
    Figure~\ref{fig:online-setting} shows that query history improves both Discovery Exact and SQL execution accuracy on BIRD and Spider. On BIRD, the gains are 5.79 and 4.13 percentage points, respectively. These gains suggest that historical queries provide
    useful context.
}
{
    \textbf{Additional results}. We present additional experimental results, such as additional details with statistical significance, predicted relationship quality when such ground-truth exists, and studies on the benefit of incorporating table-to-code relationships, in our technical report~\cite{full} in the interest of space.
}

\section{Related work}
\label{sec:related}

\stitle{Dataset annotation using LLMs.}
In addition to earlier work that uses traditional methods to expand the semantics of tables and columns~\cite{sherlock,sato,doduo, cortez2015annotating, auto-type}, recent work such as AutoDDG~\cite{autoddg} and TACO~\cite{taco} use LLMs to generate dataset and column descriptions based on table content. \method  addresses a different limitation: independently annotating each table
does not explain relationships between tables. We therefore augment per-table descriptions with structured,
lake-wide relationships and develop tools that let agents retrieve and traverse them, which are shown to improves the quality and efficiency of downstream data tasks.

\stitle{Join discovery and search.}
There is a large literature on \emph{join discovery} that aims to identify all joins from a set of tables, where most prior work focus on join discovery in smaller-scale, database settings, often with tens of tables~\cite{zhang2010foreign, papenbrock2015binder, chen2014fast, rostin2009machine}. Scalable \emph{joinable-table search} has also been studied  using efficient value-overlap methods~\cite{zhu2019josie}, where the goal is to find joinable tables for \emph{one} target table.  There is limited study thus far on scaling \emph{all-pair join discovery}  to large data lakes with tens of thousands of tables, while efficiently producing semantically plausible join relationships.
Our setting additionally requires identifying and documenting semantically plausible PK/FK relationships across the lake, where value overlap alone may produce spurious connections. 



\stitle{Table union search.} There is a long and fruitful line of research on table union search~\cite{nargesian2018table, fan2023starmie, khatiwada2023santos, bogatu2020d3l, hu2023autotus}, which retrieves tables whose columns are semantically compatible to one target search table.
While we use schema-aligned table partitions as a reference implementation for union-able relationships in this work, our focus is on documenting these relationships for downstream agents; alternative union-search methods could be easily plugged into \method using the same relationship interface.


\stitle{Relationship-based data discovery}. Systems like Aurum~\cite{fernandez2018aurum} construct and maintain a lake-wide graph of schema similarity, content similarity, and key relationships. Data.Readme shares the principle of precomputing lake-wide knowledge, but focuses on developing an  explanatory documentation framework specifically designed for LLM agents, with new relationships such as confusable relationships, and extensible to new relationship types.


\vspace{-2mm}
\section{Conclusions and Future Work}
We presented \method, structured documentation that helps LLM agents navigate
large data lakes using offline relationship extraction. 
Our experiments show the promise of \method as reusable, relationship-aware documentation that can
serve as an effective interface between data lakes and LLM agents. 

\FloatBarrier





\bibliographystyle{ACM-Reference-Format}
\bibliography{Data.Readme}
\clearpage

    


\end{document}